\documentclass[]{amsart}
\usepackage[usenames]{color}

\usepackage[normalem]{ulem}
\usepackage{amsmath,amssymb,amsthm,amsfonts,enumerate,url,mathbbol,mathrsfs,tikz,graphicx,pifont}
\usepackage{hyperref}
\usepackage{fullpage}

\usepackage[utf8]{inputenc}
\usepackage{pgfplots}
\usepgfplotslibrary{polar}

\theoremstyle{plain}
\newtheorem{theorem}{Theorem}[section]
\newtheorem{lemma}[theorem]{Lemma}
\newtheorem{proposition}[theorem]{Proposition}
\newtheorem{corollary}[theorem]{Corollary}

\theoremstyle{definition}
\newtheorem{definition}[theorem]{Definition}
\newtheorem{example}[theorem]{Example}

\newtheorem{remark}[theorem]{Remark}

\numberwithin{equation}{section}
\numberwithin{figure}{section}
\DeclareMathOperator{\Real}{Re}
\renewcommand{\Re}{\Real}
\DeclareMathOperator{\Imag}{Im}
\renewcommand{\Im}{\Imag}

\DeclareMathOperator{\sign}{sgn}

\DeclareMathOperator{\lin}{lin}
\providecommand{\abs}[1]{\left\lvert#1\right\rvert}
\providecommand{\norm}[1]{\left\lVert#1\right\rVert}

  \def\mG{\mathsf{G}}

  \def\mE{\mathsf{E}}

 \def\mv{\mathsf{v}}
 \def\me{\mathsf{e}}

\newcommand{\R}{\mathbb{R}}

\newcommand{\C}{\mathbb{C}}

\title{Airy-type evolution equations on star graphs} 



\author[D.~Mugnolo]{Delio Mugnolo}

\address{Delio Mugnolo, Lehrgebiet Analysis, Fakult\"at Mathematik und
informatik, Fern\-Universit\"at in Hagen, D-58084 Hagen, Germany}
\email{delio.mugnolo@fernuni-hagen.de}

\author[D.~Noja]{Diego Noja}

\address{Diego Noja, Dipartimento di Matematica e Applicazioni, Universit\`a di
Milano Bicocca, Via Cozzi 55, 20125 Milano, Italia }
\email{diego.noja@unimib.it}

\author[C.~Seifert]{Christian Seifert}

\address{Christian Seifert, Ludwig-Maximilians-Universit\"at M\"unchen,
Mathematisches Institut, Theresienstra{\ss}e 39, D-80333 M\"unchen, Germany}
\email{christian.seifert@mathematik.uni-muenchen.de}

\thanks{Version of \today}

\date{\today}

\begin{document}

\begin{abstract}
In the present paper the Airy operator on star graphs is defined and studied. 
The Airy operator is a third order differential operator arising in different
contexts, but our main concern is related to its role as the linear part of the
Korteweg-de Vries equation, usually studied on a line or a half-line. 
The first problem treated and solved is its correct definition, with different
characterizations, as a skew-adjoint operator on a star graph, a set of lines
connecting at a common vertex representing, for example, a network of branching channels. 
A necessary condition turns out to be that the graph is balanced, i.e. there is
the same number of ingoing and outgoing edges at the vertex. 
The simplest example is that of the line with a point interaction at the vertex.
In these cases the Airy dynamics is given by a
unitary or isometric (in the real case) group. In particular the analysis
provides the complete classification of boundary conditions giving momentum (i.e., $L^2$-norm of the solution) preserving evolution on the graph.
A second more general problem here solved is the characterization of conditions under which the Airy operator generates a contraction semigroup. 
In this case unbalanced star graphs are allowed. In both unitary and contraction dynamics, restrictions on admissible boundary conditions occur if conservation of mass (i.e.,  integral of the solution) is further imposed.
The above well posedness results can be considered preliminary to the analysis of nonlinear wave propagation on branching structures. 

\end{abstract}

\maketitle 
\section{Introduction}
We consider the partial differential equation
\begin{equation}\label{airy}
\frac{\partial u}{\partial t}=\alpha \frac{\partial^3 u}{\partial
x^3}+\beta\frac{\partial u}{\partial x}\ ,
\end{equation}
where $\alpha\in \mathbb R\setminus\{0\}$ and $\beta \in \mathbb R$, on half-bounded
intervals ($(-\infty,0)$ or $(0,\infty)$) and more generally on collections of copies thereof building those structures nowadays commonly known as \emph{metric star graphs}.

A metric star graph in the present setting is the structure represented by the set (see Figure \ref{fig:balanced_star_graph})
\[
\mE:=\mE_-\cup \mE_+
\]
where $\mE_+$ and $\mE_-$ are finite or countable collections of semi-infinite edges $\me$ parametrized by $(-\infty,0)$ or $(0,\infty)$ respectively. 
The half lines are connected at a vertex $\mv$, where suitable boundary conditions have to be imposed in order to result in a well posed boundary initial value problem. From a mathematical point of view the problem consists in a system of $|\mE_-|+|\mE_+|$ partial differential equations of the form \eqref{airy}, 
with possibly different coefficients $\alpha$ and $\beta$, coupled through the boundary condition at the vertex. 
Our main concern in this paper is exactly the characterization of boundary conditions yielding a well posed dynamics for equation \eqref{airy} on said metric star graph.
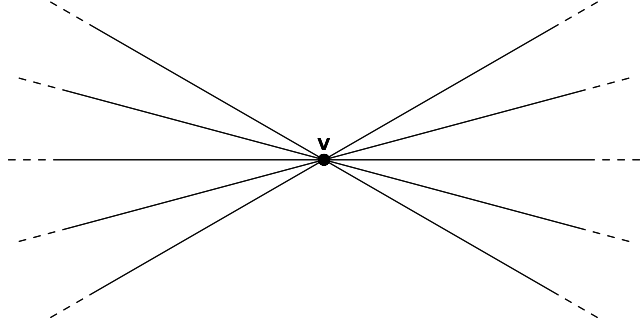
\begin{figure}[t]
\begin{center}
\begin{tikzpicture}[scale=1.0]
\foreach \x in {0,15,30,-15,-30,180,195,210,165,150}{
\draw[fill] (\x:0cm) -- (\x:3.6cm);
\draw[fill] (\x:3.7cm) -- (\x:3.8cm);
\draw[fill] (\x:3.9cm) -- (\x:4cm);
\draw[fill] (\x:4.1cm) -- (\x:4.2cm);
\draw[fill] (0:0cm) circle (2pt) node[above]{$\mv$};
}
\end{tikzpicture}
\end{center}
\caption{A balanced star graph with $|\mE_-|=|\mE_+|=5$ edges.}
\label{fig:balanced_star_graph}
\end{figure}

This is a well known problem when a Schr\"odinger operator is considered on the graph, in which case the system is called {\it quantum graph}. In this case,  extended literature exists on the topic, both for the elliptic problem and for some evolution equations, like heat or wave or reaction-diffusion equations (see \cite{BerKu12, Mugnolo14} and references therein). The analysis has been recently extended also to the case of nonlinear Schr\"odinger equation, 
in particular as regards the characterization of ground states and standing waves (see \cite{Noja14} and references therein for a review). Another dispersive nonlinear equation, the BBM equation, is treated in \cite{BonaCascaval08,MugnoloRault14}.
A partly numerical analysis and partly theoretical analysis of some special cases of the Korteweg-de Vries equation on a metric graph is given in \cite{SobAkUe15a,SobAkUe15b} and \cite{SobAkKaJa15b}; apart from these last three papers, not much seems to be known for~\eqref{airy}, the
linear part of the KdV equation. In this context, the solution of equation \eqref{airy} represents, in the long wave or small amplitude limit, the deviation of the free surface of water from its mean level in the presence of a flat bottom.
We will refer to equation \eqref{airy} as the {\it Airy equation} and to the operator on its right hand side as the {\it Airy} operator; its connection with the KdV equation is one of our main motivations for the study of this problem.
The equation~\eqref{airy} appeared for the first time in a paper by Stokes in 1847 as a contribution to the understanding of solitary waves in shallow water channels observed by Russel (\cite{Stokes47} and  \cite{Lannes13}), but some of its solutions were discussed previously by Airy (ironically, to refuse the existence of solitary waves).

One of the most fascinating features of this PDE is that it has both
dispersive and smoothing character 
and in fact it is
well known that~\eqref{airy} is governed by a
group of bounded
linear operators whenever its space domain is the real line,
cf.~\cite{CraGoo90}. \\
On the line, after Fourier transform, the solution of \eqref{airy} is given by

\[
u(t,x)=\frac{1}{2\pi}\int_{\mathbb R}\int_{\mathbb R}e^{ik[(x-y)-\alpha k^2
-\beta] t}u_0(y)\ dy\ dk\ ,
\]
where $u_0$ is the initial data of~\eqref{airy}.
The above Fourier formula can be further managed to give in
the standard $\alpha<0$ case
\begin{equation}
u(t,x)= K_t \ast u_0(x)\ , \ \ \ \ \ \ \ K_t(x)=\frac{1}{\sqrt[3]{3|\alpha|
t}}\text{Ai}\left(\frac{x-\beta t}{\sqrt[3]{3|\alpha| t}}\right)\ ,
\end{equation}
where 
\begin{equation}
\text{Ai}(x):=\frac{1}{2\pi}\int_{\R} e^{iy x +i\frac{1}{3}y^3}
dy=\frac{1}{\pi}\int_{\R} \cos\left(y x +\frac{1}{3}y^3\right) dy
\end{equation}
is the so called Airy function (to be intended with a hidden exponential convergence factor).\\ 
Translation invariance allows one to get rid of the first order term just changing to moving coordinates in the equation or directly in the solution, and it is not restrictive to put $\beta=0.$
In particular one has that
$||K_t||_{L^{\infty}(\R)}={\mathcal O}(\frac{1}{\sqrt[3]t})$, from which the
typical dispersive behavior of Airy equation solutions on the line follows. This estimate with its consequences is of outmost importance both at the linear level and in the analysis of nonlinear perturbation of 
the equation, such as the KdV equation (see \cite{LinPon14} and references therein).
\par
Much harder is the analysis of the properties of the Airy operator when translation invariance is broken, for example in the case of a half line or in the presence of an external potential.
In the first case much work has been done in the context of the analysis of well posedness of KdV equation.  The first problem one is faced with is the unidirectional character of the propagation, which requires some care in the definition of the correct boundary value problem. 
It is well known that for the standard problem ($\alpha<0$) of the KdV equation two boundary conditions at $0$ are needed on $(-\infty,0)$ and only one suffices on $(0,+\infty)$. The issues in the definition of the correct boundary conditions of KdV equation are inherited by the linear part, the Airy equation. A more complete analysis of the problem on the half line is given in the following section.\\
The Airy operator with an external potential, representing the effect of an obstacle in the propagation or the result of a linearization around known stationary solutions is studied in \cite{Miller97}, especially as regards spectral and dispersion properties, while a different model with an inhomogeneous dispersion obtained through the introduction of a space dependent coefficient for the third derivative term is studied in \cite{CraGoo90}. \\
In this paper we are interested in the generalization of the half-line example and from now on we will discuss only the case of a star graph in which the Airy operator has constant coefficients on every single edge (but possibly different coefficients from one edge to another) and discard any other source of inhomogeneity. \\
Our first goal is to properly define the Airy operator as an unbounded operator on a certain Hilbert space, in such a way it turns out to be the generator of a $C_0$-semigroup. We will consider  two cases in increasing order of generality. The first is the one in which the generated dynamics is unitary (in the complex case) or isometric (in the real case); the second is the case in which the generated dynamics is given by a contraction $C_0$-semigroup.
The easiest way to understand the nature of the problem is to consider the densely defined closable and skew-symmetric Airy operator $
u\mapsto A_0 u:= \alpha \frac{d^3 u}{d x^3}+\beta\frac{d u}{d x}$ with domain $C^\infty_c((-\infty,0))$ or $
C^\infty_c((0,+\infty))$ on the Hilbert spaces 
$L^2((-\infty,0))$ or $L^2((0,+\infty))$ respectively. Due to the fact that the Airy operator is of odd order, changing the sign of  $\alpha$ it is equivalent to exchange the positive and negative half line and so we can take $\alpha>0$ without loss of generality.\\
An easy control shows that the deficiency indices of the two operators are $(2,1)$ in the first case, $(1,2)$ in the second case, so that $iA_0$ does not have any self-adjoint extension on any of the two half lines. However, the direct sum of the operators on the two half lines results in a symmetric operator on $C^\infty_c((-\infty,0)\cup
(0,+\infty))$ with equal deficiency indices $(3,3)$. Hence, thanks to the classical von Neumann-Krein theory it admits a nine parameter family of self-adjoint extensions generating a unitary dynamics in $L^2(\R)$. Of course, this is not the only possibility to generate a dynamics: for instance, it could be the case that a suitable extension $A$ of the operator on the half line generates a non-unitary semigroup, so not conserving $L^2$-norm, that still consists of  $L^2$-{\it contractions}. 
According to the Lumer-Phillips theorem this holds true if and only if both $A$ and its adjoint $A^*$ are dissipative. Dissipativity in fact occurs {\it if the right number of correct boundary conditions are added}; for example in the standard case with $\alpha<0$, the Dirichlet condition on $u(0)$ on the half line $(0,+\infty)$ is sufficient, but two conditions, for example the two Dirichlet and Neumann boundary conditions on $u(0)$ and $u_x(0)$, are needed on $(-\infty,0)$.\\
This is the ultimate reason why more or less explicit results about Airy semigroup formulae exist in the above examples and they can be fruitfully applied to some cases of nonlinear perturbation such as the KdV equation on the half line. \\ The above basic remark, which seemingly went unnoticed in the previous literature on the subject, is the starting point to the treatment of the Airy operator on the more general case of a star graph. 
To efficiently treat the case of star graphs, we exploit the fact that the Airy operator is anti-symmetric, and we want to give first existence conditions for its skew-adjoint extensions and their classification. 
This can be done in principle in several ways and here we rely on a recent analysis making use (in the intermediate steps of the construction) of  Krein spaces with indefinite inner products, recently developed in \cite{SchSeiVoi15}. 
As suggested by the example of the half-line a necessary condition for skew-adjointness is that $\mE_-= \mE_+$, i.e. the number of incoming half 
lines is the same of outgoing half-lines. When this condition is met the graph is said to be {\it balanced}. A similar necessary condition was shown to be true in the case of the quantum momentum operator $-i\frac{d}{dx}$ on a graph (see  \cite{Carlson99, Exner13}). The complete characterization of skew-adjoint boundary conditions is more complex, and is given in Theorem \ref{thm:characskewself} and Theorem \ref{thm:characselforthog}. 
To explain, we introduce the space of boundary values at the vertex for the domain element of the adjoint operator $A_0^*$: 
These are given by $((u(0-),u'(0-),u''(0-))^T$ and $((u(0+),u'(0+),u''(0+))^T$, spanning respectively spaces $\mathcal{G}_-$ and $\mathcal{G}_+$ (notice that $u(0-), u(0+)$ etc.\ are vectors with components given by the boundary values on the single edges and ``minus'' and ``plus'' mean that they are taken on edges in $\mE_-$ or $\mE_+$ respectively). 
The boundary form of the operator $A_0$ is given by
\begin{align*}
(A_0^* u \mid v) + (u\mid A_0^*v) &=  \bigg(\begin{pmatrix} -\beta_- & 0 & -\alpha_- \\ 0 & \alpha_- & 0 \\
-\alpha_- & 0 & 0\end{pmatrix}\begin{pmatrix}
u(0-)\\u'(0-)\\u''(0-)\end{pmatrix} \bigg| \begin{pmatrix}
v(0-)\\v'(0-)\\v''(0-)\end{pmatrix} \bigg)_{\mathcal{G}_-}\\
  & - \bigg(\begin{pmatrix} -\beta_+ & 0 & -\alpha_+ \\ 0 & \alpha_+ & 0 \\
-\alpha_+ & 0 & 0\end{pmatrix}\begin{pmatrix}
u(0+)\\u'(0+)\\u''(0+)\end{pmatrix} \bigg| \begin{pmatrix}
v(0+)\\v'(0+)\\v''(0+)\end{pmatrix} \bigg)_{\mathcal{G}_+}\\ &= \bigg(B_-\begin{pmatrix}
u(0-)\\u'(0-)\\u''(0-)\end{pmatrix}\bigg| \begin{pmatrix}
u(0-)\\u'(0-)\\u''(0-)\end{pmatrix}\bigg)_{\mathcal{G}_-}-\bigg(B_+\begin{pmatrix}
u(0+)\\u'(0+)\\u''(0+)\end{pmatrix}\bigg| \begin{pmatrix}
u(0+)\\u'(0+)\\u''(0+)\end{pmatrix}\bigg)_{\mathcal{G}_+}\end{align*}
where, with obvious notation, $\alpha_\pm$ and $\beta_\pm$ are  vector-valued coefficients in $\mE_\pm$ of the Airy equation on the graph. The block matrices $B_\pm$ are non degenerate, symmetric but indefinite and endow the boundary spaces $\mathcal{G}_\pm$ with the structure of a Krein space. Correspondingly, the space $\mathcal{G}_-\oplus \mathcal{G}_+$ is endowed with the sesquilinear form
$$
\omega((x,y),(u,v)\bigr) = 
  \Bigg( \begin{pmatrix} B_- & 0 \\ 0 & -B_+\end{pmatrix} \begin{pmatrix}
x\\y\end{pmatrix} \Bigg| \begin{pmatrix} u\\v\end{pmatrix}\Bigg)\ .
$$
The first important characterization is that skew-adjoint extensions of $A_0$ are parametrized by the subspaces $X$ of $\mathcal{G}_-\oplus \mathcal{G}_+$  which are $\omega$-{\it self-orthogonal} ($X=X^{\bot}$ where orthogonality is with respect to the indefinite sesquilinear form $\omega $, see Definition \ref{defi:krein} for more details). An equivalent, more explicit parametrization is through relations between boundary values. Consider a linear operator $L:\mathcal{G}_-\to\mathcal{G}_+$ (here we describe the simplest case in which $\mathcal{G}_\pm$ are finite dimensional; for the general case see Section \ref{sec:star_graphs})
and define 
\begin{equation*}
  D(A_L)  = \Bigl\{u\in D(A_0^*);\; 
  L\bigl(u(0-),u'(0-),u''(0-)\bigr) = \bigl(u(0+),u'(0+),u''(0+)\bigr)\Bigr\},\ \ \ 
  A_L u  = -A_0^*u.
\end{equation*}
Then the skew adjoint extensions of $A_0$ are the operators $A_L$ such that $L$ is \emph{$(\mathcal{G}_-,\mathcal{G}_+)$-unitary} ($\langle Lx \mid Ly \rangle_+ = \langle x\mid y\rangle_- $ and again see Definition \ref{defi:krein}).

Along similar lines can be treated the characterization of extensions $A$ of $A_0$ generating a contraction semigroup.
According to Lumer-Phillips theorem, one has that $A$ and $A^*$ have to be dissipative. For a linear operator $L: \mathcal{G}_- \to \mathcal{G}_+$ define $A_L\subseteq -A_0^*$ as above.
Denoting by $L^\sharp$ the $(\mathcal{G}_+,\mathcal{G}_-)$-adjoint of the operator $L$ (cf.\ Definition~\ref{defi:krein}) one  has that
$A_L$ is the generator of a contraction semigroup
  if and only if $L$ is a $(\mathcal{G}_-,\mathcal{G}_+)$-contraction (i.e., $\langle Lx \mid Lx \rangle_+ \leq \langle x\mid x\rangle_- $ for all $x\in \mathcal G_-$)  and $L^{\sharp}$ is a $(\mathcal{G}_+,\mathcal{G}_-)$-contraction (i.e., $\langle L^{\sharp}x \mid L^{\sharp} x \rangle_- \leq \langle x\mid x\rangle_+$ for all $x\in \mathcal G_+$).\\ This condition allows for a different number of ingoing and outgoing half-lines, and it is well adapted to study realistic configurations such as a branching channel.\\Notice that in both the skew-adjoint or dissipative case one can consider the possible conservation of an additional physical quantity, the {\it mass}, coinciding with the integral of the domain element of $A$. Conservation of mass  (see Remark \ref{rem:mass}) is characterized by boundary conditions satisfying the constraint
\begin{equation*}
\sum_{\me\in\mE_-} \alpha_\me u_\me''(0-) - \sum_{\me\in\mE_+} \alpha_\me u_\me''(0+) + \sum_{\me\in\mE_-} \beta_\me u_\me(0-) - \sum_{\me\in\mE_+} \beta_\me u_\me (0+) = 0\ .
\end{equation*}
Of course the requirement of mass conservation restricts the allowable boundary conditions, both in the skew-adjoint and the dissipative case.\\
In Section 4 we provide a collection of more concrete examples. Already at the level of general analysis previously done it is clear that a distinguished class of boundary conditions exists, in which the first derivative boundary values $u'(0\pm)$ are separated. This means that they do not interact with the boundary values of $u$ and $u''$ and satisfy the transmission relation $u'(0+)=Uu(0-)$ with $U$ unitary in the skew-adjoint case and a contraction in the dissipative case, while $u(0\pm)$ and $u''(0\pm)$ are coupled. This case is described at the beginning of Section 4. Some more special examples deserve an interest. The first is the graph consisting of two half-lines. This case is interesting also because it can be interpreted as describing the presence of an obstacle or point interaction on the line. 
In the skew-adjoint case this perturbation does not destroy the conservation of momentum during the evolution. Moreover, there are both skew-adjoint and more generically contractive boundary conditions that 
conserve the mass as well. In this sense the two half-lines case corresponds to a forcing or interaction which however can preserve in time some physical quantity.  A related recent analysis is given in \cite{DeconinckSheilsSmith16}  where an inhomogeneous interface problem for the Airy operator on the line (in the special case $\beta=0$) is studied by means of the Fokas Unified Transform Method (UTM) (see \cite{Fokas08}) and necessary and sufficient conditions for its solvability are given in terms of interface conditions (in that paper the authors prefer to distinguish between interface and boundary conditions; here this usage is not followed). It surely would be interesting to compare the interface conditions studied in the quoted paper with the ones derived in the present one both as regards the skew-adjoint case and the contraction case.\\
A second class of examples are given for the graph with three half-lines, which falls necessarily in the non skew-adjoint case. 
A more accurate analysis of this example is relevant in perspective, because of the possible application to the analysis of flow in branching channels, which has attracted some attentions in recent years (see \cite{NachbinSimoes12, NachbinSimoes15} and references therein, and the interesting early paper \cite{Jakovkis91} where different models of flows are treated but much of the analysis has a general value). 
It is not at all clear which boundary conditions should be the correct ones from a physical point of view, and the complete classification of those giving generation at least for the Airy operator is a first step to fully understand realistic situations. In general one expects, on theoretical and experimental 
grounds, that one dimensional reductions retain some memory of the geometry (for example the angles of the fork) which are not contemplated in the pure graph description, and a further step in the analysis should consists in including these effects through the introduction of further phenomenological parameters or additional terms in the equation.\\

The paper is so organized. In Section 2 the simplest case of a single half-line is considered, giving also some comments on the previous literature on the subject. In Section 3 the complete construction of the Airy operator on a star graph in the skew-adjoint and dissipative case is given. The exposition is self contained as regards the preliminary definitions on graphs and operators on Krein spaces. Besides construction, some general properties of the Airy operators are further studied. In Section 4 some examples are treated.
   
  \section{The case of the half-line}
As a warm-up in the present Section we consider the Airy equation on the half-line. This is an often considered subject, because of its relevance in connection with the analogous initial boundary value problem for the KdV equation, which turns out to be a rather challenging problem especially when studied in low Sobolev regularity (see for pertinent papers on the subject with information on the linear problem \cite{BonaSunZhang01, CollianderKenig02, HayashiKaikina04, Holmer06, Faminskii07, FokasHiMan16} and references therein). For the interesting case of the interval see also \cite{ColGhi01,BonaSunZhang03}). We recall that Korteweg and de Vries derived, under several hypotheses, the equation
\begin{equation}\label{eq:kdv}
\frac{\partial \eta}{\partial t}=\frac{3}{2}\sqrt{\frac{g}{\ell}}\left(\frac{\sigma}{3}\frac{\partial^3 \eta}{\partial x^3} +\frac{2\alpha}{3}\frac{\partial \eta}{\partial x}+\frac12 \frac{\partial \eta^2}{\partial x}\right)
\end{equation}
where the unknown $\eta$ is the elevation of the water surface with respect to its average depth $\ell$ in a shallow canal (see~\cite{KordeV95} for an explanation of the parameters $g,\sigma,\alpha$; $\sigma$ is implicitly assumed to be negative there), or \cite{Lannes13}) for a modern and complete analysis. Renaming coefficients one obtains the KdV equation with parameter $\alpha$ in front of the third derivative and $\beta$ in front of the first derivative. The physically relevant case of a semi-infinite channel represented by half-line $[0,+\infty)$ and a
wavemaker placed at $x=0$ is described, if dissipation is neglected, by the KdV equation in which the linear part has $\alpha<0$ and $\beta<0$ (see \cite{BonaSunZhang01} and references therein). So one obtains, neglecting nonlinearity, the Airy equation on the positive half-line as the linear part of the KdV equation {\it with the above sign of coefficients}. The first derivative term disappear on the whole line changing to a moving reference system, but on the half-line or on a non translational invariant set should be retained, and we do this. The sign of coefficient of the third derivative is important and interacts with the choice of the half-line, left or right, where propagation occurs. Depending on the sign of $\alpha$, a
different number of boundary conditions has to be imposed: one and two boundary conditions,
on $(0,\infty)$ and $(-\infty,0)$ respectively if $\alpha<0$, and vice versa two and one boundary conditions if $\alpha>0$. This reflects the fact that the partial differential equation~\eqref{airy} has unidirectional nature, and is explained in the literature in several different ways. An explanation making resort to the behavior of characteristic curves is recalled for example in \cite{DeconinckSheilsSmith16}, where the authors however notice that it has only an heuristic value. A more convincing brief discussion is given, for $\beta=0$ as the previous one, in \cite{Holmer06}, which we reproduce now with minor modifications and considering $\alpha=\pm 1$ for simplicity, which can always be achieved by rescaling the space variable.
On the left half line, after multiplying both sides of
\eqref{airy} by $u$ and integrating on $(-\infty,0)$ one obtains 
\begin{equation*}
\int^0_{-\infty} u_t(t,x)u(t,x) \ dx = \ \pm u_{xx}(t,0-)u(t,0-) \mp
\frac{1}{2}u_x^2(t,0-  +\beta \frac{1}{2} u^2(t,0-).
\end{equation*}
Integrating in time on $(0,T)$ one finally obtains the identity
\begin{equation*}
\begin{aligned}
&\frac{1}{2}\int_{-\infty}^0 u^2(T,x)\ dx - \frac{1}{2}\int_{-\infty}^0
u^2(0,x)\ dx \\ &= \pm \int_0^T u_{xx}(t,0-)u(t,0-)\ dt \mp \frac{1}{2}\int_0^T
u_{x}^2(t,0-)\ dt +\beta \frac{1}{2}\int_0^T u^2(t,0-)\ dt.
\end{aligned}
\end{equation*}
If we consider the operator with $\alpha=-1$, which is the standard Airy operator,
we conclude that the boundary condition $u(t,0-)=0$ alone and the initial
condition $u(0,x)=0$ are compatible with a non vanishing solution of the
equation: to force a vanishing solution one has to fix the boundary value
$u_x(t,0-)=0$ also. So uniqueness is guaranteed by both boundary conditions on
$u$ and $u_x$. On the contrary the above identity for the operator with $\alpha=1$
with the boundary condition $u(t,0-)=0$ and the initial condition $u(0,x)=0$
imply $u(T,x)=0$ for any $T$. So one has uniqueness in the presence of the only
condition on $u(t,0)$.  \\
In the case of the positive half line $(0,\infty)$ one obtains from \eqref{airy}
\begin{equation*}
\begin{aligned}
&\frac{1}{2}\int_0^{\infty} u^2(T,x)\ dx-\frac{1}{2}\int^{\infty}_0 u^2(0+,x)\
dx \\= &\mp \int_0^T u_{xx}(t,0+)u(t,0+)\ dt \pm \frac{1}{2}\int_0^T
u_{x}^2(t,0+)\ dt - \beta \frac{1}{2}\int_0^T u^2(t,0+)\ dt
\end{aligned}
\end{equation*}
and for $\alpha=-1$ one has uniqueness in the presence of the boundary
condition for $u(t,0+)$ alone, while for  $\alpha=+1$ uniqueness needs the
specification of both $u(t,0+)$ and $u(t,0+)$.
Considering the difference of two solutions corresponding to identical initial
data and boundary conditions, one concludes that the same properties hold true
in the case of general non vanishing boundary conditions.

One can go in greater depth and the following lemma is the point of departure. It clarifies some of the properties of the Airy operator on the half-line and will be extended, with many consequences, to a star graph in the following sections.
\begin{lemma}\label{lem:oper-1edge}
Consider the operator 
\[
u\mapsto H_0 u:= \alpha \frac{d^3 u}{d x^3}+\beta\frac{d u}{d x}
\]
with domain
\[
\hbox{either }\quad C^\infty_c(-\infty,0)\quad\hbox{or }\quad
C^\infty_c(0,+\infty)
\]
on the Hilbert space 
\[
L^2(-\infty,0),\quad\hbox{resp. }\quad L^2(0,+\infty).
\]
Then $iH_0$ is densely defined, closable and symmetric. However, its deficiency
indices are $(2,1)$ in the first case, $(1,2)$ in the second case, so $iH_0$
does not have any self-adjoint extensions in either case.
\end{lemma}

\begin{proof}
One checks directly that the adjoint of $H_0$ on $L^2(-\infty,0)$ and
$L^2(0,\infty)$ is $H_0^*=-\alpha \frac{d^3 u}{d x^3}-\beta\frac{d u}{d x}$ with
domain $H^3(-\infty,0)$ and $H^3(0,\infty)$, respectively.

We are going to solve the elliptic problem $(iH_0^*\mp i{\rm Id})u=0$ in
$L^2(-\infty,0)$ and $L^2(0,\infty)$. Let us consider the sign $+$ and discuss
\[
\alpha \frac{d^3 u}{d x^3}+\beta\frac{d u}{d x}+ u=0
\]
for $u\in L^2(-\infty,0)$ or $L^2(0,\infty)$ without any boundary conditions. A
tedious but elementary computation shows that a general solution is a linear
combination of complex exponentials of the form
\[\begin{split}
u(x)&=C_1{\rm e}^{-x\frac{ \left( i \sqrt{3} \left(  \left( 12\, \sqrt{3}
\sqrt{{\frac {4\,{\beta}^{3}+27\,\alpha}{\alpha}}}-108 \right) {\alpha}^{2}
\right) ^{2/3}+12\,i \sqrt{3}\alpha\,\beta+ \left(  \left( 12\, \sqrt{3}
\sqrt{{\frac {4\,{\beta}^{3}+27\,\alpha}{\alpha}}}-108 \right) {\alpha}^{2}
\right) ^{2/3}-12\,\beta\,\alpha \right)}{12\alpha\left( \left( 12\, \sqrt{3}
\sqrt{{\frac {4\,{\beta}^{3}+27\,\alpha}{\alpha}}}-108 \right)
{\alpha}^{2}\right)^\frac{1}{3}}}\\
&\quad +C_2{\rm e}^{x\frac{ \left( i \sqrt{3} \left(  \left( 12\, \sqrt{3}
\sqrt{{\frac {4\,{\beta}^{3}+27\,\alpha}{\alpha}}}-108 \right) {\alpha}^{2}
\right) ^{2/3}+12\,i \sqrt{3}\alpha\,\beta- \left(  \left( 12\, \sqrt{3}
\sqrt{{\frac {4\,{\beta}^{3}+27\,\alpha}{\alpha}}}-108 \right) {\alpha}^{2}
\right) ^{2/3}+12\,\beta\,\alpha \right)}{12\alpha\left( \left( 12\, \sqrt{3}
\sqrt{{\frac {4\,{\beta}^{3}+27\,\alpha}{\alpha}}}-108 \right)
{\alpha}^{2}\right)^\frac{1}{3}}}\\
&\quad+C_3\,{{\rm e}^{x\frac{ \left(  \left(  \left( 12\, \sqrt{3} \sqrt{{\frac
{4\,{\beta}^{3}+27\,\alpha}{\alpha}}}-108 \right) {\alpha}^{2} \right)
^{2/3}-12\,\beta\,\alpha \right)}{6\alpha \left( \left( 12\, \sqrt{3}
\sqrt{{\frac {4\,{\beta}^{3}+27\,\alpha}{\alpha}}}-108 \right)
{\alpha}^{2}\right)^\frac{1}{3}}}}
\end{split}
\]
for general constants $C_1,C_2,C_3$. However, carefully checking the real parts
of the exponents one deduces that such functions are square integrable on
$(-\infty,0)$ and $(0,\infty)$ if and only if $C_3=0$ and $C_1=C_2=0$,
respectively, thus yielding the claim on the deficiency indices. The remaining
case can be treated likewise.
\end{proof}
This shows that the Airy operator cannot be extended to a skew adjoint operator generating a unitary dynamics in $L^2$. 
Moreover in some sense the Airy operator ``with the wrong sign'' or too a few boundary conditions has too much spectrum to allow for uniqueness (see \cite[Theorem 23.7.2]{HiPhi57}).
This is not however the whole story, and one can obtain more precise information and some more guiding ideas renouncing to a unitary evolution and asking simply for generation of a contractive semigroup. 
To this end we consider the Lumer-Phillips condition and its consequences
(see  for example \cite[Corollary 3.17]{EngelNagel99}).
Again, considering only the case $\alpha=\pm 1$ which is enough, one has by integration by parts
\begin{equation*}
\begin{aligned}
\int_{-\infty}^0(\pm u_{xxx} +\beta u_x) v \ dx&=\int_{-\infty}^0 u(\mp v_{xxx}
-\beta v_x)  \ dx\\
&\pm u_{xx}(0)v(0) +\beta u(0)v(0)\mp u_x(0)v_x(0) \pm u(0)v_{xx}(0)
\end{aligned}
\end{equation*}
The operator $H_{\pm,\beta}$ with domain $${\mathcal D}(H_{\pm,\beta})=\left\{
u\in H^3(-\infty,0)\ ;\ u(0)=0\ , \ \ u_x(0)=0 \right\}$$ and action
$$H_{\pm,\beta}u=\pm u_{xxx}+\beta u_x$$
satisfies
\begin{equation*}
\langle H_{\pm,\beta}u,u \rangle= 0
\end{equation*}
and it is dissipative (in fact conservative).\\
The operator $H^\lozenge_{\pm,\beta}$ with domain $${\mathcal
D}(H^\lozenge_{\pm,\beta})=\left\{ u\in H^3(-\infty,0)\ ;\ u(0)=0\right\}$$ and
action $$H^\lozenge_{\pm,\beta}u=\mp u_{xxx}-\beta u_x$$
satisfies
\begin{equation*}
\langle H^\lozenge_{\pm,\beta}u,u \rangle= \pm \frac{1}{2}\bigl(u_x(0-)\bigr)^2
\end{equation*}
and so $H^\lozenge_{+,\beta}$ is accretive and $H^\lozenge_{-,\beta}$ is
dissipative for every $\beta\in \R$. \\
$H_{\pm,\beta}$ and $H^\lozenge_{\pm,\beta}$ are in fact adjoint one to
another: 
\begin{equation}
H^\lozenge_{\pm,\beta}=H^*_{\pm,\beta}\ .
\end{equation}
In particular this means that $H_{+,\beta}$ and $H^*_{+,\beta}$ are both
accretive and so they do not generate a continuous contraction semigroup in
$L^2(-\infty, 0)$. On the contrary, $H_{-,\beta}$ and $H^*_{-,\beta}$ are both
dissipative and generate a contraction semigroup in $L^2(-\infty,0).$
With our convention of writing of the Airy equation, this gives well posedness
on $(-\infty, 0)$ for the standard Airy equation \eqref{airy} with two boundary conditions (generator $H_{-,\beta}$).
The specular situation occurs for $(0,\infty)$, exchanging the roles and
definitions of $H$ and $H^\lozenge$, and one has that $H^\lozenge_{+,\beta}$ generates
a contraction semigroup on $L^2(0,\infty)$ and the standard Airy equation \eqref{airy} with a single boundary condition is well posed.


\section{The case of a metric star graph}
\label{sec:star_graphs}
Star graphs can be regarded as the building blocks of more complicated graphs;
for the purpose of investigating (local) boundary conditions, they are
sufficiently generic. Therefore, in this section we are going to develop the
theory of the counterpart of the operator $H_0$ defined on a star graph
$\mG$, which indeed turns out to display some unexpected behaviours in
comparison with its simpler relative introduced in
Lemma~\ref{lem:oper-1edge}.

Upon replacing an interval $(-\infty,0)$ by $(0,\infty)$ or vice versa, we may
assume all coefficients $\alpha$ to have the same sign on each edge $\me$ of the
star graph. Throughout this paper we are going to follow the convention that all coefficients are positive.

\begin{proposition}\label{prop:first}
Consider a quantum graph consisting of finitely or countably many halflines
$\mE:=\mE_-\cup \mE_+$ and let $(\alpha_\me)_{\me\in \mE},(\beta_\me)_{\me\in \mE}$ be two sequences of real numbers with $\alpha_\me>0$ for all $\me \in \mE$.
Consider the operator $A_0$ defined by
\[
\begin{split}
  D(A_0) & := \bigoplus_{\me\in \mE_-} C^\infty_c(-\infty,0)\oplus
\bigoplus_{\me\in \mE_+} C^\infty_c(0,+\infty),\\
A_0&:(u_\me)_{\me\in \mE} \mapsto \left(\alpha_\me \frac{d^3 u_\me}{d
x^3}+\beta_\me\frac{d u_\me}{d x}\right)_{\me\in \mE}.
\end{split}
\]
Then $iA_0$ is densely defined and symmetric on the Hilbert space
$$L^2(\mG):=\bigoplus_{\me\in \mE_-} L^2(-\infty,0)\oplus
\bigoplus_{\me\in \mE_+} L^2(0,+\infty) $$
and its defect indices are
$(2|\mE_-|+|\mE_+|,|\mE_-|+2|\mE_+|)$. Accordingly, $A_0$ has skew-self-adjoint
extensions on $L^2(\mG)$ if and only if $|\mE_+|=|\mE_-|$.
\end{proposition}

In order to avoid some technicalities we will assume that the sequences $(\alpha_\me)_{\me\in \mE}$ and $(\beta_\me)_{\me\in \mE}$ are bounded, and furthermore, that
$(\tfrac{1}{\alpha_\me})_{\me\in\mE}$ is bounded as well.

Unlike in the case of the Laplace operator, and in spite of the relevance of
related physical models, like the KdV equation, there seems to be no canonical
or natural choice of boundary conditions to impose on~\eqref{airy} on a star
graph. For this reason, we are going to characterize all boundary conditions
within certain classes.
Since~\eqref{airy} plays a role in dispersive systems in which conservation of
the initial data's norm is expected, we are going to focus on those extensions
that generate unitary groups, or isometric semigroups, or at least contractive semigroups.

\subsection{Extensions of $A_0$ generating unitary groups}

By Stone's theorem, generators of unitary groups are exactly the
skew-self-adjoint operators.
In order to determine the skew-self-adjoint extensions of $A_0$, take 
\[
u,v\in D(A_0^*)=\bigoplus_{\me\in \mE_-} H^3(-\infty,0)\oplus \bigoplus_{\me\in
\mE_+} H^3(0,+\infty).
\]

\begin{remark}
Le $u\in D(A_0^*)$.
Note that by Sobolev's Lemma and the boundedness assumption on the $(\alpha_\me)$ and $(\beta_\me)$ we obtain $(u^{(k)}_\me(0-))_{\me\in\mE_-}\in \ell^2(\mE_-)$ and 
$(u^{(k)}_\me(0+))_{\me\in\mE_+}\in \ell^2(\mE_+)$ for $k\in\{0,1,2\}$.
\end{remark}

Following the classical extension theory, cf.~\cite[Chapt.~3]{Sch12}, we write
down the boundary form to obtain
\[
\begin{split}
&(A_0^*u\mid v)+(u\mid A_0^* v)=\\
&\qquad= -\sum_{\me\in \mE_-}\int_{-\infty}^0 \left(\alpha_\me u'''_\me
+\beta_\me u'_\me \right)\overline{v_\me} \ dx - \sum_{\me\in
\mE_+}\int_0^{+\infty} \left(\alpha_\me  u'''_\me +\beta_\me u'_\me
\right)\overline{v_\me} \ dx\\
&\qquad\qquad - \sum_{\me\in \mE_-}\int_{-\infty}^0 \overline{\left(\alpha_\me 
v'''_\me +\beta_\me v'_\me \right)}u_\me \ dx - \sum_{\me\in
\mE_+}\int_0^{+\infty}\overline{ \left(\alpha_\me  v'''_\me +\beta v'_\me
\right)}u_\me \ dx\\
&\qquad= - \sum_{\me\in \mE_-}\alpha_\me  u''_\me
(x)\overline{v_\me(x)}\Big|_{-\infty}^0 
+\sum_{\me\in \mE_-}\int_{-\infty}^0 \alpha_\me  u''_\me \overline{ v'_\me }\
dx\\
&\qquad\qquad - \sum_{\me\in \mE_-}\alpha_\me \overline{ v''_\me
(x)}u_\me(x)\Big|_{-\infty}^0 + \sum_{\me\in \mE_-}\int_{-\infty}^0 \alpha_\me 
u'_\me \overline{ v''_\me }\ dx\\
&\qquad\qquad - \sum_{\me\in \mE_+}\alpha_\me  u''_\me (x)\overline{v_\me
(x)}\Big|_0^{+\infty}
+ \sum_{\me\in \mE_+}\int_0^{+\infty} \alpha_\me  u''_\me \overline{ v'_\me }\
dx\\
&\qquad\qquad - \sum_{\me\in \mE_+}\alpha_\me \overline{ v''_\me (x)}u_\me
(x)\Big|_0^{+\infty} +\sum_{\me\in \mE_+}\int_0^{+\infty} \alpha_\me  u'_\me
\overline{ v''_\me }\ dx\\
&\qquad\qquad - \sum_{\me\in \mE_-}\int_{-\infty}^0 \beta_\me u'_\me
\overline{v_\me} \ dx
 - \sum_{\me\in \mE_-}\int_{-\infty}^0 \beta_\me\overline{ v'_\me }u_\me \ dx -
\sum_{\me\in \mE_+}\int_0^{+\infty} \beta u'_\me \overline{v_\me} \ dx 
 - \sum_{\me\in \mE_+}\int_0^{+\infty}\beta_\me\overline{ v'_\me }u_\me \ dx\\
&\qquad= - \sum_{\me\in \mE_-}\alpha_\me  u''_\me
(x)\overline{v_\me(x)}\Big|_{-\infty}^0 
+ \sum_{\me\in \mE_-}\alpha_\me  u'_\me (x)\overline{ v'_\me
(x)}\Big|_{-\infty}^0\\
&\qquad\qquad - \sum_{\me\in \mE_-}\alpha_\me \overline{ v''_\me
(x)}u_\me(x)\Big|_{-\infty}^0 + \sum_{\me\in \mE_+}\alpha_\me  u'_\me
(x)\overline{ v'_\me (x)}\Big|_0^{-\infty}\\
&\qquad\qquad - \sum_{\me\in \mE_+}\alpha_\me  u''_\me (x)\overline{v_\me
(x)}\Big|_0^{+\infty}
\\
&\qquad\qquad- \sum_{\me\in \mE_+}\alpha_\me \overline{ v''_\me (x)}u_\me
(x)\Big|_0^{+\infty}\\
&\qquad\qquad
- \sum_{\me\in \mE_-}\beta_\me  u_\me(x)\overline{ v_\me(x)}\Big|_{-\infty}^0 
- \sum_{\me\in \mE_+}\beta_\me  u_\me(x)\overline{ v_\me(x)}\Big|_0^{+\infty}\\
&\qquad= - \sum_{\me\in \mE_-}\alpha_\me  u''_\me (0)\overline{v_\me(0)} +
\sum_{\me\in \mE_+}\alpha_\me  u''_\me (0)\overline{v_\me(0)}
- \sum_{\me\in \mE_-}\alpha_\me \overline{ v''_\me (0)}u_\me(0) + \sum_{\me\in
\mE_+}\alpha_\me \overline{ v''_\me (0)}u_\me(0)\\
&\qquad\qquad - \sum_{\me\in \mE_-}\beta_\me  u_\me(0)\overline{ v_\me(0)}
+ \sum_{\me\in \mE_+}\beta_\me  u_\me(0)\overline{ v_\me(0)}
+ \sum_{\me\in \mE_-}\alpha_\me  u'_\me (0)\overline{ v'_\me (0)}- \sum_{\me\in
\mE_+}\alpha_\me  u'_\me (0)\overline{ v'_\me (0)}.
\end{split}
\]

Thus, abbreviating $u(0-):=(u_\me(0-))_{\me\in \mE_-}$ and similarly for the other
terms, and identifying with an abuse of notation $\alpha_{\pm}$ and $\beta_{\pm}$ with the corresponding
multiplication operator, i.e.
\begin{equation}\label{eq:aee}
\begin{split}
\alpha_\pm x&:=(\alpha_\me x_\me)_{\me\in \mE_\pm},\qquad x\in \ell^2(\mE_\pm),\\
\beta_\pm x&:=(\beta_\me x_\me)_{\me\in \mE_\pm},\qquad x\in \ell^2(\mE_\pm),
\end{split}
\end{equation}
we obtain
\begin{equation}\label{eq:finally}
\begin{split}
  (A_0^* u \mid v) + (u\mid A_0^*v) & = -(\alpha_- u''(0-)\mid v(0-)) +
(\alpha_+ u''(0+)\mid v(0+)) \\
  & - (\alpha_- u(0-)\mid v''(0-)) + (\alpha_+ u(0+)\mid v''(0+)) \\
  & - (\beta_- u(0-)\mid v(0-)) + (\beta_+ u(0+)\mid v(0+)) \\
  & + (\alpha_- u'(0-)\mid v'(0-)) - (\alpha_+ u'(0+)\mid v'(0+)) \\
  & = \bigg(\begin{pmatrix} -\beta_- & 0 & -\alpha_- \\ 0 & \alpha_- & 0 \\
-\alpha_- & 0 & 0\end{pmatrix}\begin{pmatrix}
u(0-)\\u'(0-)\\u''(0-)\end{pmatrix} \bigg| \begin{pmatrix}
v(0-)\\v'(0-)\\v''(0-)\end{pmatrix} \bigg)_{\mathcal{G}_-}\\
  &\quad - \bigg(\begin{pmatrix} -\beta_+ & 0 & -\alpha_+ \\ 0 & \alpha_+ & 0 \\
-\alpha_+ & 0 & 0\end{pmatrix}\begin{pmatrix}
u(0+)\\u'(0+)\\u''(0+)\end{pmatrix} \bigg| \begin{pmatrix}
v(0+)\\v'(0+)\\v''(0+)\end{pmatrix} \bigg)_{\mathcal{G}_+},
\end{split}
\end{equation}
where
\[
\mathcal{G}_-:=\ell^2(\mE_-)\oplus \ell^2(\mE_-)\oplus \ell^2(\mE_-)
\]
and
\[
\mathcal{G}_+:=\ell^2(\mE_+)\oplus\ell^2(\mE_+)\oplus\ell^2(\mE_+)\ .
\]
(We stress the difference from $\mG$, which we let denote the quantum graph.)
Consider on the graph $G(A_0^*)$ of the operator $A_0^*$ a linear and surjective
operator $F\colon G(A_0^*)\to \mathcal{G}_-\oplus \mathcal{G}_+$ defined by
\begin{equation}\label{eq:Fdef}
F((u,A_0^*u)):=
\Bigl(\bigl(u(0-),u'(0-),u''(0-)\bigr),\bigl(u(0+),u'(0+),u''(0+)\bigr) \Bigr).
\end{equation}

Following the terminology in~\cite[Examples 2.7]{SchSeiVoi15}, let $\Omega$ be the standard symmetric form on $G(A_0^*)$, i.e.
\[\Omega\bigl((u,A_0^*u),(v,A_0^*v)\bigr):= \bigg(\begin{pmatrix} 0 & 1 \\ 1 & 0
\end{pmatrix} \begin{pmatrix} u\\A_0^*u \end{pmatrix} \bigg| \begin{pmatrix}
v\\A_0^*v\end{pmatrix} \bigg)_{L^2(\mG)}, \qquad (u,A_0^*u),(v,A_0^*v)\in G(A_0^*),\]
and define a sesquilinear form $\omega\colon \mathcal{G}_-\oplus \mathcal{G}_+
\times \mathcal{G}_-\oplus \mathcal{G}_+ \to \mathbb{C}$
by
\begin{equation}\label{eq:omega}
\omega((x,y),(u,v)\bigr) := 
  \Bigg( \begin{pmatrix} B_- & 0 \\ 0 & -B_+\end{pmatrix} \begin{pmatrix}
x\\y\end{pmatrix} \Bigg| \begin{pmatrix} u\\v\end{pmatrix}\Bigg)_{ \mathcal{G}_-\oplus \mathcal{G}_+},
  \end{equation}
where $B_\pm$ is the linear block operator matrix on $\mathcal{G}_\pm$ defined by
\[B_\pm(x,x',x'') := \begin{pmatrix} -\beta_\pm & 0 & -\alpha_\pm \\ 0 &
\alpha_\pm & 0 \\ -\alpha_\pm & 0 & 0\end{pmatrix}\begin{pmatrix}
x\\x'\\x''\end{pmatrix}\ ,\qquad  (x,x',x'')\in \mathcal{G}_\pm.
\]
Observe that neither $B_+$ nor $B_-$ are definite operators.

Then~\eqref{eq:finally} can be re-written as
\begin{equation}\label{eq:omegaOmega}
\Omega\bigl((u,A_0^*u),(v,A_0^*v)\bigr) = \omega\bigl(F(u,A_0^*u),
F(v,A_0^*v)\bigr) \quad\hbox{for all }u,v\in D(A_0^*)\ .
\end{equation}

\begin{remark}
  Note that, since $\beta_\pm$,
  $\alpha_\pm$ and $\frac{1}{\alpha_\pm}$ are bounded, we have $B_\pm\in
  \mathcal{L}(\mathcal{G}_\pm)$, $B_\pm$ are injective and $B_\pm^{-1}\in \mathcal{L}(\mathcal{G}_\pm)$.
\end{remark}

Our method is based on the notion of Krein space, i.e., of a vector space
endowed with an \emph{indefinite} inner product. 

\begin{definition}
  Define an (indefinite) inner product $\langle\cdot\mid \cdot\rangle_\pm\colon
\mathcal{G}_\pm\times\mathcal{G}_\pm\to\mathbb{K}$ by
\[\langle x\mid y\rangle_\pm:=(B_\pm x \mid y),\qquad x,y\in \mathcal G_\pm\ .\]
\end{definition}

Then $(\mathcal{G}_\pm, \langle\cdot\mid \cdot\rangle_\pm)$ are Krein spaces and
$\langle\cdot\mid \cdot\rangle_\pm$ is non-degenerate, i.e.\ for $x\in
\mathcal{G}_\pm$ with $\langle x \mid x \rangle_\pm = 0$ it follows $x=0$.

\begin{remark}
  Let $\mathcal{K}$ be a vector space and $\langle\cdot\mid \cdot\rangle$ an (indefinite) inner product on $\mathcal{K}$ such that $(\mathcal{K}, \langle\cdot\mid \cdot\rangle)$ is a Krein space.
  Then there exists an inner product $(\cdot\mid\cdot)$ on $\mathcal{K}$ such that $(\mathcal{K},(\cdot\mid\cdot))$ is a Hilbert space.
  Notions such as closedness or continuity are then defined by the underlying Hilbert space structure.
\end{remark}

\begin{definition}\label{defi:krein}
Let $\mathcal{K}_-$, $\mathcal{K}_+$ be Krein spaces, $\omega\colon \mathcal{K}_-\oplus \mathcal{K}_+ \times \mathcal{K}_-\oplus \mathcal{K}_+\to \C$ sesquilinear.
\begin{enumerate}
\item A subspace $X$ of $\mathcal{K}_-\oplus \mathcal{K}_+$ is called
\emph{$\omega$-self-orthogonal} if
\[X=X^{\bot_\omega}:= \{(x,y)\in \mathcal{K}_-\oplus \mathcal{K}_+;\;
\omega((x,y),(u,v)) = 0\hbox{ for all }(u,v)\in X\}\ .\]
\item  Given a densely defined linear operator $L$ from $\mathcal{K}_-$ to
$\mathcal{K}_+$, its  \emph{$(\mathcal{K}_-,\mathcal{K}_+)$-adjoint}
$L^{\sharp}$ is 
\[
\begin{split}
D(L^{\sharp}) &:= \{y\in \mathcal{K}_+;\; \exists z\in \mathcal{K}_-\, \:
\langle Lx \mid y \rangle_+ = \langle x\mid z\rangle_-\hbox{ for all } x\in
D(L)\}\\
L^{\sharp}y&:=z\ .
\end{split}
\]
Clearly, $L^\sharp$ is then a linear operator from $\mathcal{K}_+$ to $\mathcal{K}_-$.
\item   A linear operator $L$ from $\mathcal{K}_-$ to $\mathcal{K}_+$ is called a
\emph{$(\mathcal{K}_-,\mathcal{K}_+)$-contraction} if
  \[\langle Lx \mid Lx \rangle_+ \leq \langle x\mid x\rangle_- \quad\hbox{for all
}x\in D(L)\ .\]

\item   A linear operator $L$ from $\mathcal{K}_-$ to $\mathcal{K}_+$ is called
\emph{$(\mathcal{K}_-,\mathcal{K}_+)$-unitary} if $D(L)$ is dense, $R(L)$ is
dense, $L$ is injective, and finally $L^{\sharp} = L^{-1}$.  
\end{enumerate}
\end{definition}
If in particular $\mathcal K_-,\mathcal K_+$ are  Hilbert spaces, then obviously $(\mathcal{K}_-,\mathcal{K}_+)$-adjoint/contraction/unitary operators are nothing but the usual objects of Hilbert space operator theory.

Note that if $L$ is a $(\mathcal{K}_-,\mathcal{K}_+)$-unitary, then
\begin{equation}\label{eq:defunit}
\langle Lx \mid Ly \rangle_+ = \langle x\mid y\rangle_- \quad\hbox{for all
}x,y\in D(L).
\end{equation}
We stress that unitary operators between Krein spaces need not be bounded.

By \cite[Corollary 2.3 and Example 2.7(b)]{SchSeiVoi15}
we can now characterize skew-self-adjoint extensions $A$
of $A_0$ -- i.e., skew-self-adjoint restrictions of $A_0^*$ -- and therefore
self--adjoint extensions of $iA_0$. 

\begin{theorem}\label{thm:characskewself}
An extension $A$ of $A_0$ is skew-self-adjoint if and only if there exists an
$\omega$-self-orthogonal subspace $X\subseteq \mathcal{G}_-\oplus \mathcal{G}_+$
for which $G(A) = F^{-1}(X)$, where $F$ is the operator defined in~\eqref{eq:Fdef} and $G(A)$ is the graph of $A$.
\end{theorem}

Hence, $\omega$-self-orthogonal subspaces $X$ parametrize the skew-self-adjoint
extensions $A$ of $A_0$. A more explicit description of these objects is given
next.

\begin{theorem}\label{thm:characselforthog}
  Let $X\subseteq \mathcal{G}_-\oplus \mathcal{G}_+$ be a subspace.
  Then $X$ is $\omega$-self-orthogonal if and only if there exists a
$(\mathcal{G}_-,\mathcal{G}_+)$-unitary operator $L$ such that $X = G(L)$.
\end{theorem}

\begin{proof}
  First, let $X$ be $\omega$-self-orthogonal.
  For $(x,y)\in X$ we have
  \[\omega\bigl((x,y),(x,y)\bigr) =0,\]
  i.e.
  \[\langle y \mid y \rangle_+ = \langle x\mid x\rangle_-.\]
  For the special case $x=0$ we obtain $\langle y\mid y\rangle_+ = 0$. Since
$\langle \cdot\mid\cdot\rangle_+$ is non-degenerate, 
  it follows that $y=0$.
  Thus, $X$ is the graph of a linear operator $L$ from $\mathcal{G}_-$ to
$\mathcal{G}_+$.
  For the special case $y=0$ we obtain $\langle x\mid x\rangle_- = 0$. Since
$\langle \cdot\mid\cdot\rangle_-$ is non-degenerate, 
  it follows that $x=0$. Thus, $L$ is injective.
  
  For $x,y\in D(L)$, we have $(x,Lx),(y,Ly)\in X$, so
  \[\omega\bigl((x,Lx),(y,Ly)\bigr) = 0,\]
  i.e.
  \[\langle Lx \mid Ly \rangle_+ = \langle x\mid y\rangle_-.\]
  Let $x\in D(L)^\bot$. Then
  \[(x\mid y) = 0 \quad(y\in D(L)).\]
  Thus,
  \[0 = (x\mid y) = (B_- B_-^{-1}x \mid y) = \langle B_-^{-1}x \mid y\rangle_- =
\omega\bigl((B_-^{-1}x,0),(y,Ly)\bigr) \quad(y\in D(L)).\]
  Hence, $(B_-^{-1}x,0) \in X^{\bot_\omega} = X = G(L)$, so $B_-^{-1}x = 0$, and
therefore $x=0$. Thus, $L$ is densely defined. Similarly, we obtain that $R(L)$
is dense.
  
  For $x\in D(L)$, $z\in R(L)$ we have
  \[\langle Lx \mid z \rangle_+ = \langle x\mid L^{-1}z \rangle_-.\]
  Thus, $R(L)\subseteq L^{\sharp}$, and $L^{\sharp} z = L^{-1}z$ for all $z\in
R(L)$, i.e.\ $L^{-1}\subseteq L^{\sharp}$.
  Let $(y,x)\in G(L^{\sharp})$. Then
  \[\langle Lz \mid y \rangle_+ = \langle z\mid x \rangle_- \quad(z\in D(L)),\]
  i.e.
  \[\omega\bigl((z,Lz), (x,y)\bigr) = 0 \quad(z\in D(L)).\]
  Hence, $(x,y)\in X^{\bot_\omega} = X = G(L)$, and therefore $(y,x)\in
G(L^{-1})$. Therefore, $L^{\sharp} = L^{-1}$, so $L$ is a
$(\mathcal{G}_-,\mathcal{G}_+)$-unitary.
    
  \medskip
  
  Let $L$ be a $(\mathcal{G}_-,\mathcal{G}_+)$-unitary operator from
$\mathcal{G}_-$ to $\mathcal{G}_+$, $X=G(L)$. Let $x\in D(L)$. Then
  \[\langle Lx \mid Ly \rangle_+ = \langle x\mid y\rangle_- \quad(y\in D(L)),\]
  i.e.
  \[\omega\bigl((x,Lx), (y,Ly)\bigr) = 0 \quad(y\in D(L)).\]
  Thus, $(x,Lx)\in X^{\bot_\omega}$, and therefore $X\subseteq
X^{\bot_\omega}$. 
  
  Let now $(z,y)\in X^{\bot_\omega}$. Then, for $x\in D(L)$, we have
  \[\omega\bigl((z,y), (x,Lx)\bigr) = 0,\]
  and thus
  \[\langle Lx \mid y \rangle_+ = \langle x\mid z\rangle_-.\]
  By definition of $L^{\sharp}$, we obtain $y\in D(L^{\sharp})$ and $L^{\sharp}y
= z$. Since $L^\sharp = L^{-1}$, we find $(z,y)\in G(L) = X$. 
  Hence, $X^{\bot_\omega} \subseteq X$.
  
  Combining both parts, we see that $X$ is $\omega$-self-orthogonal.
\end{proof}

\begin{remark}\label{rem:mass}
If Theorem~\ref{thm:characskewself} applies, then Stone's theorem immediately yields that the Airy equation~\eqref{airy} on the quantum star graph $\mG$ is governed by a unitary group acting on $L^2(\mG)$, hence it has a unique solution $u\in C^1(\mathbb R;L^2(\mG))\cap C(\mathbb R;D(A))$ that continuously depends on the initial data $u_0\in L^2(\mG)$. 
Because the group is unitary, the momentum $\|u\|^2_{L^2(\mG)}$ is conserved as soon as we can apply Theorem~\ref{thm:characskewself}.
By the above computation we also see that
\begin{equation}\label{eq:mass}
\begin{split}
\frac{\partial }{\partial t}\int_{\mG} u(t,x)dx&=
\sum_{\me\in \mE_-}\int_{-\infty}^0 \alpha_\me u'''_\me (t,x)+\beta_\me u'_\me(t,x)dx+
\sum_{\me\in \mE_-}\int_0^{\infty}\alpha_\me u'''_\me (t,x)+\beta_\me u'_\me(t,x)dx\\
&=  \sum_{\me\in \mE_-}\alpha_\me  u''_\me
(0-)-  \sum_{\me\in \mE_+}\alpha_\me  u''_\me (0+)+ \sum_{\me\in \mE_-}\beta_\me  u_\me(0-)- \sum_{\me\in \mE_+}\beta_\me  u_\me(0+)\ .
\end{split}
\end{equation}
In other words, the solution of the system enjoys conservation of mass -- just like the solution to the classical Airy equation on $\mathbb R$ -- if and only if additionally
\begin{equation}\label{eq:yesmass}
\sum_{\me\in\mE_-} \alpha_\me u_\me''(0-) - \sum_{\me\in\mE_+} \alpha_\me u_\me''(0+) + \sum_{\me\in\mE_-} \beta_\me u_\me(0-) - \sum_{\me\in\mE_+} \beta_\me u_\me (0+) = 0\ .
\end{equation}
\end{remark}

\begin{remark}
  Note that, since $\beta_\pm, \alpha_\pm,\frac{1}{\alpha_\pm}$ are bounded,
the form $\omega$ introduced in~\eqref{eq:omega} is continuous. Thus,
$\omega$-self-orthogonal subspaces are closed, so the corresponding
$(\mathcal{G}_-,\mathcal{G}_+)$-unitary $L$ is closed.
  We do not know whether $L$ is in fact continuous (this holds true in Hilbert
spaces, but we are not aware of any argument in Krein spaces).
\end{remark}

\begin{remark}
As already remarked at the beginning of this section, star graphs can be seen as generic building blocks of quantum graphs. Apart from their interest in scattering theory, star graphs whose edges are semi-infinite still display all relevant features for the purpose of studying extensions of operators on compact graphs. 
Indeed, our analysis is essentially of variational nature and therefore it only depends on the orientation of the edges and the boundary values of a function on the graph. It is therefore clear that analogous results could be formulated for graphs that include edges of finite length, too, like the \textit{flower graph} depicted in Figure~\ref{fig:flower}.
\begin{figure}
\begin{tikzpicture}[scale=0.7]
  \begin{polaraxis}[grid=none, axis lines=none]
     \addplot[mark=none,domain=0:360,samples=300] { abs(cos(7*x/2))};
   \end{polaraxis}
   \draw[fill] (3.42,3.42) circle (2pt);
 \end{tikzpicture}
\caption{A flower graph on seven edges.}
\label{fig:flower}
\end{figure}
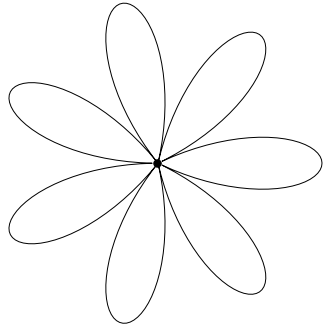
Clearly, an interesting feature of flower graphs is that they are automatically balanced, i.e., the number of incoming and outgoing edges from the (only) vertes is equal: accordingly, the operator $A_0$ on a flower graph always admits skew-self-adjoint extensions.
\end{remark}

\subsection{Extensions of $A_0$ generating contraction semigroups}

Let $A$ be an extension of $A_0$ such that $A\subseteq -A_0^*$.
We focus on generating contraction semigroups. By the Lumer-Phillips Theorem and corollaries of it we have to show that $A$ and $A^*$ are dissipative.
Since we are dealing with Hilbert spaces, $A$ is dissipative if and only if $\Re (Au\mid u)\leq 0$ for all $u\in D(A)$, and analogously for $A^*$. 
Let $L$ be a densely defined linear operator from $\mathcal{G}_-$ to $\mathcal{G}_+$. Then we define $A_L\subseteq -A_0^*$ by 
\[G(A_L) = F^{-1}(G(L)),\]
i.e.
\begin{align*}
  D(A_L) & = \Bigl\{u\in D(A_0^*);\; \bigl(u(0-),u'(0-),u''(0-)\bigr)\in D(L), \\
  & \qquad\qquad\qquad\qquad L\bigl(u(0-),u'(0-),u''(0-)\bigr) = \bigl(u(0+),u'(0+),u''(0+)\bigr)\Bigr\},\\
  A_L u & = -A_0^*u.
\end{align*}

\begin{lemma}
\label{lem:A_L dissipative}
  Let $L$ be a densely defined linear operator from $\mathcal{G}_-$ to $\mathcal{G}_+$. Then $A_L$ is dissipative if and only if
  $L$ is a $(\mathcal{G}_-,\mathcal{G}_+)$-contraction.
\end{lemma}

\begin{proof}
  Let $u\in D(A_L)$. Then
  \begin{align*}
    -2\Re (A_L u \mid u) & = \Omega\bigl((u,-A_Lu),(u,-A_Lu)\bigr) \\
    & = \bigl\langle \bigl(u(0-),u'(0-),u''(0-)\bigr), \bigl(u(0-),u'(0-),u''(0-)\bigr)\bigr\rangle_- \\
    & - \bigl\langle L\bigl( u(0-),u'(0-),u''(0-)\bigr), L\bigl(u(0-),u'(0-),u''(0-)\bigr)\bigr\rangle_+.
  \end{align*}
  Thus, $A_L$ is dissipative if and only if
  \begin{align*}
    & \bigl\langle L\bigl( u(0-),u'(0-),u''(0-)\bigr), L\bigl(u(0-),u'(0-),u''(0-)\bigr)\bigr\rangle_+ \\
    &\qquad \leq \bigl\langle \bigl(u(0-),u'(0-),u''(0-)\bigr), \bigl(u(0-),u'(0-),u''(0-)\bigr)\bigr\rangle_-  \quad\hbox{for all }u\in D(A_L).
  \end{align*}
  By definition of $A_L$ we have $D(L) = \bigl\{\bigl(u(0-),u'(0-),u''(0-)\bigr); u\in D(A_L)\bigr\}$, so the assertion follows.
\end{proof}

Analogously, we obtain a characterisation for dissipativity of the adjoint $A_L^*$ of $A_L$.
Here and in the following, $L^\sharp$ denotes the $(\mathcal{G}_-,\mathcal{G}_+)$-adjoint of the operator $L$, cf.\ Definition~\ref{defi:krein}.

\begin{lemma}
  Let $L$ be a densely defined linear operator from $\mathcal{G}_-$ to $\mathcal{G}_+$. Then
  \begin{align*}
    D(A_L^*) & = \Bigl\{u\in D(A_0^*);\; \bigl(u(0+),u'(0+),u''(0+)\bigr) \in D(L^\sharp), \\
    & \qquad\qquad\qquad\qquad L^\sharp \bigl(u(0+),u'(0+),u''(0+)\bigr) = \bigl(u(0-),u'(0-),u''(0-)\bigr)\Bigr\},\\
    A_L^* u & = A_0^* u.
  \end{align*}
\end{lemma}

\begin{proof}
  Let $u\in D(A_L)$, $v\in D(A_0^*)$. Then
  \begin{align*}
    (A_L u \mid v) & = (u \mid A_0^*v) - \bigl\langle \bigl(u(0-),u'(0-),u''(0-)\bigr) \mid \bigl(v(0+),v'(0+),v''(0+)\bigr)\bigr\rangle_- \\
    & + \bigl\langle L\bigl(u(0-),u'(0-),u''(0-)\bigr),\bigl(v(0+),v'(0+),v''(0+)\bigr)\bigr\rangle_+.
  \end{align*}
  Hence, $v\in D(A_L^*)$ if and only if $\bigl(v(0+),v'(0+),v''(0+)\bigr)\in D(L^\sharp)$ and 
  \[L^\sharp\bigl(v(0+),v'(0+),v''(0+)\bigr) = \bigl(v(0-),v'(0-),v''(0-)\bigr),\]
  and then $A_L^* v = A_0^* v$.  
\end{proof}

\begin{lemma}
\label{lem:A_L^*_dissipative}
  Let $L$ be a densely defined linear operator from $\mathcal{G}_-$ to $\mathcal{G}_+$. Then $A_L^*$ is dissipative if and only if
  $L^\sharp$ is a $(\mathcal{G}_+,\mathcal{G}_-)$-contraction.
\end{lemma}

\begin{proof}
  Let $u\in D(A_L^*)$. Then
  \begin{align*}
    2\Re (A_L^* u \mid u) & = \Omega((u,A_L^*u),(u,A_L^*u)) \\
    & = \bigl\langle L^\sharp \bigl( u(0+),u'(0+),u''(0+)\bigr), L^\sharp \bigl(u(0+),u'(0+),u''(0+)\bigr)\bigr\rangle_- \\
    & - \bigl\langle \bigl(u(0+),u'(0+),u''(0+)\bigr), \bigl(u(0+),u'(0+),u''(0+)\bigr)\bigr\rangle_+\ .
  \end{align*}
  Thus, $A_L^*$ is dissipative if and only if
  \begin{align*}
    & \bigl\langle L^\sharp \bigl( u(0+),u'(0+),u''(0+)\bigr), L^\sharp \bigl(u(0+),u'(0+),u''(0^+)\bigr)\bigr\rangle_- \\
    &\qquad \leq \bigl\langle \bigl(u(0+),u'(0+),u''(0+)\bigr), \bigl(u(0+),u'(0+),u''(0+)\bigr)\bigr\rangle_+  \quad\hbox{for all }u\in D(A_L)\ .
  \end{align*}
  By definition of $A_L$ we have $D(L^\sharp) = \bigl\{\bigl(u(0+),u'(0+),u''(0+)\bigr); u\in D(A_L)\bigr\}$, so the assertion follows.
\end{proof}

\begin{theorem}
\label{thm:contraction_semigroup}
  Let $L$ be a densely defined linear operator from $\mathcal{G}_-$ to $\mathcal{G}_+$. Then $A_L$ is the generator of a contraction semigroup
  if and only if $L$ is a $(\mathcal{G}_-,\mathcal{G}_+)$-contraction and $L^\sharp$ is a $(\mathcal{G}_+,\mathcal{G}_-)$-contraction.
\end{theorem}

\begin{proof}
  Let $A_L$ generate a $C_0$-semigroup $(T(t))_{t\ge0}$ of contractions. By the Lumer-Phillips Theorem, $A_L$ is dissipative.
  Since $A_L^*$ generates the semigroup $T^*$ defined by $T^*(t):=T(t)^*$ ($t\geq 0$), which is also a $C_0$-semigroup of contractions, the Lumer-Phillips Theorem assures that $A_L^*$ is dissipative as well.
  Then the Lemmas \ref{lem:A_L dissipative} and \ref{lem:A_L^*_dissipative} yield that $L$ and $L^\sharp$ are contractions.
  
  Now, let $L$ and $L^\sharp$ be contractions.
  Then Lemma \ref{lem:A_L dissipative} and \ref{lem:A_L^*_dissipative} yields that $A_L$ and $A_L^*$ are dissipative. Hence, $A_L$ generates a contraction semigroup.
\end{proof}

\begin{remark}\label{rem:mass-2}
Like in the case of Remark~\ref{rem:mass}, 
if Theorem~\ref{thm:contraction_semigroup} applies, then the Airy equation~\eqref{airy} on the quantum star graph $\mG$ has a unique solution $u\in C^1(\mathbb R_+;L^2(\mG))\cap C(\mathbb R_+;D(A))$ that continuously depends on the initial data $u_0\in L^2(\mG)$. 
Because the $C_0$-semigroup is contractive but not unitary, the momentum $\|u\|^2_{L^2(\mG)}$ is in general not conserved, but as in Remark~\ref{rem:mass} the system does enjoy conservation of mass if and only if additionally~\eqref{eq:yesmass} holds.
\end{remark}

\subsection{Separating the first derivatives}

The special structure of $B_\pm$ suggests to separate the boundary values of the first derivative from the ones for the function and for the second derivative. 
In this case, one can describe the boundary conditions also in another (equivalent but seemingly easier) way.

Note that $\alpha_\me>0$ for all $\me\in\mE$. We will write $\ell^2(\mE_\pm,\alpha_\pm)$ for the weighted $\ell^2$-space of sequences indexed by $\mE_\pm$ with inner product given by
\[\bigl(x\mid y\bigr)_{\ell^2(\mE_\pm,\alpha_\pm)} := \sum_{\me\in\mE_\pm} x_\me\overline{y}_\me \alpha_\me = \bigl(\alpha_\pm x\mid y\bigr)\]
for all $x,y\in \ell^2(\mE_\pm,\alpha_\pm)$, which turns them into Hilbert spaces.

For $u,v\in D(A_0^*)$ we then obtain
\begin{align*}
  & (A_0^* u \mid v) + (u\mid A_0^*v) \\
  & = \bigl(\alpha_- u'(0-) \mid v'(0-) \bigr) - \bigl(\alpha_+ u'(0+) \mid v'(0+) \bigr) \\
  & \quad 
  + \bigg(\begin{pmatrix} -\beta_-  & -\alpha_- \\
	  -\alpha_- & 0 \end{pmatrix}\begin{pmatrix}
	  u(0-)\\u''(0-)\end{pmatrix} \bigg| \begin{pmatrix}
	  v(0-)\\v''(0-)\end{pmatrix} \bigg)
  - \bigg(\begin{pmatrix} -\beta_+  & -\alpha_+ \\
	  -\alpha_+ & 0 \end{pmatrix}\begin{pmatrix}
	  u(0+)\\u''(0+)\end{pmatrix} \bigg| \begin{pmatrix}
	  v(0+)\\v''(0+)\end{pmatrix} \bigg) \\
  & = \bigl(\alpha_- u'(0-) \mid v'(0-) \bigr) - \bigl(\alpha_+ u'(0+) \mid v'(0+) \bigr) \\
  & \quad 
  - \bigg( \begin{pmatrix} \alpha_- & 0 \\ 0 & -\alpha_+ \end{pmatrix}\begin{pmatrix} u''(0-) \\ u''(0+)\end{pmatrix} + \begin{pmatrix} \frac{\beta_-}{2} & 0 \\ 0 & -\frac{\beta_+}{2} \end{pmatrix}\begin{pmatrix} u(0-) \\ u(0+)\end{pmatrix} \bigg| \begin{pmatrix} v(0-) \\ v(0+)\end{pmatrix} \bigg) \\
  & \quad - \bigg( \begin{pmatrix} u(0-) \\ u(0+)\end{pmatrix} \bigg| \begin{pmatrix} \alpha_- & 0 \\ 0 & -\alpha_+ \end{pmatrix}\begin{pmatrix} v''(0-) \\ v''(0+)\end{pmatrix} + \begin{pmatrix} \frac{\beta_-}{2} & 0 \\ 0 & -\frac{\beta_+}{2} \end{pmatrix}\begin{pmatrix} v(0-) \\ v(0+)\end{pmatrix}\bigg)\ .
\end{align*}

Let $Y\subseteq \ell^2(\mE_-)\oplus \ell^2(\mE_+)$ be a closed subspace, 
$U$ a densely defined linear operator from $\ell^2(\mE_-,\alpha_-)$ to $\ell^2(\mE_+,\alpha_+)$, and consider
\begin{align*}
  D(A_{Y,U}) & := \{u\in D(A_0^*);\; (u(0-),u(0+))\in Y, \\
  & \qquad\qquad \begin{pmatrix} -\alpha_- & 0 \\ 0 & \alpha_+ \end{pmatrix}\begin{pmatrix} u''(0-) \\ u''(0+)\end{pmatrix} + \begin{pmatrix} -\frac{\beta_-}{2} & 0 \\ 0 & \frac{\beta_+}{2} \end{pmatrix}\begin{pmatrix} u(0-) \\ u(0+)\end{pmatrix} \in Y^\bot,\\
  & \qquad\qquad u'(0-)\in D(U),\, u'(0+) = U u'(0-)\},\\
  A_{Y,U} u & := -A_0^*u.					
\end{align*}

\begin{proposition}
\label{prop:adjoint}
  Let $Y\subseteq \ell^2(\mE_-)\oplus \ell^2(\mE_+)$ be a closed subspace, $U$ a densely defined linear operator from $\ell^2(\mE_-,\alpha_-)$ to $\ell^2(\mE_+,\alpha_+)$.
  Then 
  \begin{align*}
    D(A_{Y,U}^*) & = \{u\in D(A_0^*);\; (u(0-),u(0+))\in Y, \\
    & \qquad\qquad \begin{pmatrix} -\alpha_- & 0 \\ 0 & \alpha_+ \end{pmatrix}\begin{pmatrix} u''(0-) \\ u''(0+)\end{pmatrix} + \begin{pmatrix} -\frac{\beta_-}{2} & 0 \\ 0 & \frac{\beta_+}{2} \end{pmatrix}\begin{pmatrix} u(0-) \\ u(0+)\end{pmatrix} \in Y^\bot,\\
    & \qquad\qquad u'(0+)\in D(U^*),\, u'(0-) = U^* u'(0+)\},\\
    A_{Y,U}^* u & = A_0^*u.					
  \end{align*}
\end{proposition}

\begin{proof}
  Let $u\in D(A_{Y,U})$, $v\in D(A_0^*)$. Then
  \begin{align*}
    (A_{Y,U}u \mid v) 
    & = (u\mid A_0^*v) 
    + \bigl( \alpha_+ Uu'(0-) \mid v'(0+)  \bigr) - \bigl( \alpha_- u'(0-) \mid v'(0-) \bigr) \\
    & \quad + \bigg( \begin{pmatrix} \alpha_- & 0 \\ 0 & -\alpha_+ \end{pmatrix}\begin{pmatrix} u''(0-) \\ u''(0+)\end{pmatrix} + \begin{pmatrix} \frac{\beta_-}{2} & 0 \\ 0 & -\frac{\beta_+}{2} \end{pmatrix}\begin{pmatrix} u(0-) \\ u(0+)\end{pmatrix} \bigg| \begin{pmatrix} v(0-) \\ v(0+)\end{pmatrix} \bigg) \\
    & \quad + \bigg( \begin{pmatrix} u(0-) \\ u(0+)\end{pmatrix} \bigg| \begin{pmatrix} \alpha_- & 0 \\ 0 & -\alpha_+ \end{pmatrix}\begin{pmatrix} v''(0-) \\ v''(0+)\end{pmatrix} + \begin{pmatrix} \frac{\beta_-}{2} & 0 \\ 0 & -\frac{\beta_+}{2} \end{pmatrix}\begin{pmatrix} v(0-) \\ v(0+)\end{pmatrix}\bigg).
  \end{align*}
  
  Let now $v\in D(A_{Y,U}^*)$ such that $(A_{Y,U}u \mid v) = (u\mid A_{Y,U}^*v) = (u\mid A_0^*v)$.
  Choosing $u$ such that $(u(0-),u(0+)) = (u'(0-),u'(0+)) = 0$, we obtain $(v(0-),v(0+))\in Y$.
  For all $(u_-,u_+)\in Y$ there exists $u\in D(A_{Y,U})$ such that $(u(0-),u(0+)) = (u_-,u_+)$, and $(u'(0-),u'(0+)) = 0$. Thus,
  \[\begin{pmatrix} \alpha_- & 0 \\ 0 & -\alpha_+ \end{pmatrix}\begin{pmatrix} v''(0-) \\ v''(0+)\end{pmatrix} + \begin{pmatrix} \frac{\beta_-}{2} & 0 \\ 0 & -\frac{\beta_+}{2} \end{pmatrix}\begin{pmatrix} v(0-) \\ v(0+)\end{pmatrix} \in Y^\bot.\]
  Thus, we arrive at
  \[\bigl( \alpha_+ Uu'(0-) \mid v'(0+)  \bigr) - \bigl( \alpha_- u'(0-) \mid v'(0-) \bigr) = 0 \quad (u\in D(A_{Y,U})).\]
  Note that for all $x\in D(U)$ there exists $u\in D(A_{Y,U})$ such that $u'(0-) = x$.
  Hence, $v'(0+)\in D(U^*)$ and $U^* v'(0+) = v'(0-)$.  
\end{proof}

\begin{corollary}
\label{cor:skew-self-adjoint_separated}
  Let $Y\subseteq \ell^2(\mE_-)\oplus \ell^2(\mE_+)$ be a closed subspace, $U$ a densely defined linear operator from $\ell^2(\mE_-,\alpha_-)$ to $\ell^2(\mE_+,\alpha_+)$.
  Then $A_{Y,U}$ is skew-self-adjoint if and only if $U$ is unitary.
 \end{corollary}

\begin{proof}
  This is a direct consequence of Proposition \ref{prop:adjoint}.
\end{proof}

\begin{corollary}
\label{cor:A_dissipative}
  Let $Y\subseteq \ell^2(\mE_-)\oplus \ell^2(\mE_+)$ be a closed subspace, $U$ a densely defined linear operator from $\ell^2(\mE_-,\alpha_-)$ to $\ell^2(\mE_+,\alpha_+)$.
  Then $A_{Y,U}$ is dissipative if and only if $U$ is a contraction.  
\end{corollary}

\begin{proof}
  Let $u\in D(A_{Y,U})$. Then
  \[(A_{Y,U} u \mid u) = (u \mid A_0^* u) + \bigl( \alpha_+ Uu'(0-) \mid U u'(0-)  \bigr) - \bigl( \alpha_- u'(0-) \mid u'(0-) \bigr).\]
  Since $A_0^*u = -A_{Y,U} u$, we obtain
  \[ 2\Re (A_{Y,U} u \mid u) = \bigl( \alpha_+ Uu'(0-) \mid U u'(0-)  \bigr) - \bigl( \alpha_- u'(0-) \mid u'(0-) \bigr).\]
  Hence, $A_{Y,U}$ is dissipative if and only if $U$ is a contraction.
\end{proof}

%

\begin{corollary}
\label{cor:A^*_dissipative}
  Let $Y\subseteq \ell^2(\mE_-)\oplus \ell^2(\mE_+)$ be a closed subspace, $U$ a densely defined linear operator from $\ell^2(\mE_-,\alpha_-)$ to $\ell^2(\mE_+,\alpha_+)$.
  Then $A_{Y,U}^*$ is dissipative if and only if $U^*$ is a contraction.  
\end{corollary}

\begin{proof}
  Let $u\in A_{Y,U}^*$. Then, similarly as for $A_{Y,U}$, we have
  \[(A_{Y,U}^* u \mid u) = (u \mid -A_0^* u) - \bigl( \alpha_+ u'(0+) \mid  u'(0+)  \bigr) + \bigl( \alpha_- U^* u'(0+) \mid U^* u'(0+) \bigr).\]
  Since $A_0^*u = A_{Y,U}^* u$, we obtain
  \[ 2\Re (A_{Y,U} u \mid u) = \bigl( \alpha_- U^* u'(0+) \mid U^* u'(0+) \bigr) - \bigl( \alpha_+ u'(0+) \mid  u'(0+)  \bigr).\]
  Hence, $A_{Y,U}^*$ is dissipative if and only if $U^*$ is a contraction.
\end{proof}

%

\begin{theorem}
\label{thm:contraction_separated}
  Let $Y\subseteq \ell^2(\mE_-)\oplus \ell^2(\mE_+)$ be a closed subspace, $U$ a densely defined linear operator from $\ell^2(\mE_-,\alpha_-)$ to $\ell^2(\mE_+,\alpha_+)$.
  Then $A_{Y,U}$ generates a contraction $C_0$-semigroup if and only if $U$ is a contraction and $U^*$ is a contraction if and only if $U$ is a contraction.
\end{theorem}

\begin{proof}
  Let $A_{Y,U}$ generate a contraction semigroup. Then $A_{Y,U}$ is dissipative by the Lumer-Phillips Theorem, so Corollary \ref{cor:A_dissipative} yields that $U$ is a contraction.
  
  Let $U$ and $U^*$ be contractions. Then, Corollaries \ref{cor:A_dissipative} and \ref{cor:A^*_dissipative} assure that $A_{Y,U}$ and $A_{Y,U}^*$ are dissipative.
  Hence, $A_{Y,U}$ generates a semigroup of contractions.
  Clearly, $U^*$ is a contraction provided $U$ is a contraction.
\end{proof}

%

\subsection{Reality of the semigroup}

Let $(T(t))_{t\ge 0}$ be a $C_0$-semigroup on $L^2({\mG})$. We say that $(T(t))_{t\ge 0}$ is \emph{real} if $T(t)\Re u = \Re T(t) u$ for all $u\in L^2({\mG})$ and $t\geq 0$.
Put differently, a semigroup is real if it maps real-valued functions into real-valued functions.

For simplicity we will only consider the case of contractive semigroups here.

\begin{proposition}\label{prop:realsemig}
  Let $L$ be a densely defined linear operator from $\mathcal{G}_-$ to $\mathcal{G}_+$, such that $L$ and $L^\sharp$ are contractions. 
  Let $(T(t))_{t\ge0}$ be the $C_0$-semigroup generated by $A_L$. Let $L$ be real, i.e.\ for $x\in D(L)$ we have $\Re x\in D(L)$ and $\Re Lx = L\Re x$. Then $(T(t))_{t\ge0}$ is real.
\end{proposition}

\begin{proof}
  By Theorem \ref{thm:contraction_semigroup}, $(T(t))_{t\ge0}$ is a contraction semigroup.
  Let $P$ be the projection from $L^2({\mG};\C)$ to $L^2({\mG};\R)$, i.e.\ $Pu := \Re u$. By \cite[Corollary 9.6]{ISEM18} realness of the semigroup $(T(t))_{t\ge0}$ is equivalent to the condition
  \[\Re (A_L u \mid u-Pu)\leq 0 \quad\hbox{for all }u\in D(A_L),\]
  i.e.
  \[\Re (A_L u \mid i\Im u) \leq 0 \quad\hbox{for all }u\in D(A_L).\]
  Let $u\in D(A_L)$. Since $L$ is real we obtain $\Re u,\Im u\in D(A_L)$. Since $A_L v$ is real for $v\in D(A_L)$ real we have
  \[\Re (A_L(\Re u + i\Im u) \mid i\Im u) = \Im (A_L\Re u\mid \Im u) + \Re (A_L \Im u\mid \Im u) = \Re (A_L \Im u\mid \Im u) \leq 0,\]
  since $A_L$ is dissipative by Lemma \ref{lem:A_L dissipative}.  
\end{proof}

\begin{remark}
  The semigroup generated by $A_L$ is not positivity preserving, i.e., nonnegative functions need not be mapped to nonnegative functions: indeed, also in this case positivity of the semigroup is again equivalent to 
\begin{equation}\label{eq:nopositi}
\Re (A_L u \mid u-Pu)\leq 0 \quad\hbox{for all }u\in D(A_L),
\end{equation}
where $P$ is now the projection of $L^2(\mG;\mathbb C)$ onto the positive cone of $L^2(\mG;\R)$, i.e., $Pu:=(\Re u)_+$. Let $u$ be a real-valued function: integrating by parts and neglecting without loss of generality the transmission conditions (due to locality of the operator), one sees that 
  \[
  \Re (A_L u \mid u-Pu)=
  -\int_{\mG} u''' u_- dx=
  \int_{\{u\le 0\}} u''' u dx=
  -\frac12 |u'|^2\Bigg|_{\partial\{u\le 0\}}\ .
  \]
Of course, wherever an $H^3$-function changes sign its first derivative need not vanish, so condition~\eqref{eq:nopositi} cannot be satisfied.

  Analogously, the semigroup is also not $L^\infty$-contractive, i.e., the inequality $\norm{e^{tA_L} u}_\infty\leq \norm{u}_\infty$ fails for some $u\in L^2({\mG})\cap L^\infty({\mG})$ and some $t\geq 0$. In this case, the relevant projection onto the closed convex subset
  $C:=\{u\in L^2({\mG});\; \abs{u}\leq 1\}$
  of $L^2(\mG)$ is defined by $Pu:= \bigl(\abs{u}\wedge 1\bigr)\sign u$, hence 
  \[
  u-Pu:=(|u|-1)_+\sign u\ .
  \]
\end{remark}

We also obtain realness of $(T(t))_{t\ge0}$ in case of separated boundary conditions.

\begin{proposition}
  Let $Y\subseteq \ell^2(\mE_-)\oplus \ell^2(\mE_+)$ be a closed subspace, $U\colon \ell^2(\mE_-,\alpha_-)\to \ell^2(\mE_+,\alpha_+)$ linear and contractive.
  Let $(T(t))_{t\ge0}$ be the $C_0$-semigroup generated by $A_{Y,U}$.
  Assume that $(\Re x,\Re y)\in Y$ for all $(x,y)\in Y$ and
  $U$ to be real, i.e.\ $\Re Ux = U\Re x$ for all $x\in \ell^2(\mE_-,\alpha_-)$. Then $(T(t))_{t\ge0}$ is real.
\end{proposition}

\begin{proof}
  By Theorem \ref{thm:contraction_separated}, $(T(t))_{t\ge0}$ is a contraction semigroup.
  Let $P$ be the projection from $L^2({\mG};\C)$ to $L^2({\mG};\R)$, i.e.\ $Pu := \Re u$. 
  By \cite[Corollary 9.6]{ISEM18} realness of the semigroup $(T(t))_{t\ge0}$ is equivalent to
  \[\Re (A_{Y,U} u \mid u-Pu)\leq 0 \quad\hbox{for all }u\in D(A_{Y,U}),\]
  i.e.
  \[\Re (A_{Y,U} u \mid i\Im u) \leq 0 \quad\hbox{for all }u\in D(A_{Y,U}).\]
  Let $u\in D(A_{Y,U})$. The assumptions imply $\Re u\in D(A_{Y,U})$, and therefore also $\Im u\in D(A_{Y,U})$.
  Since $A_{Y,U} v$ is real for $v\in D(A_{Y,U})$ real we have
  \[\Re (A_{Y,U}(\Re u + i\Im u) \mid i\Im u) = \Im (A_{Y,U}\Re u\mid \Im u) + \Re (A_{Y,U} \Im u\mid \Im u) = \Re (A_{Y,U} \Im u\mid \Im u) \leq 0,\]
  since $A_{Y,U}$ is dissipative by Corollary \ref{cor:A_dissipative}.  
\end{proof}

\section{Examples}

\subsection{The case of two halflines}

First, let us consider the case of the real line with a singular interaction at the origin, i.e.\ $\abs{\mE_-} = \abs{\mE_+} = 1$.

\begin{figure}
\begin{tikzpicture}
\draw (-5,0) -- (5,0);
\draw[fill] (-5.1,0) -- (-5.2,0);
\draw[fill] (-5.3,0) -- (-5.4,0);
\draw[fill] (-5.5,0) -- (-5.6,0);
\draw[fill] (5.1,0) -- (5.2,0);
\draw[fill] (5.3,0) -- (5.4,0);
\draw[fill] (5.5,0) -- (5.6,0);
\draw[fill] (0,0) circle (2pt) node[above]{$0$};
\end{tikzpicture}
\caption{A graph consisting of two halflines.}\label{fig:line}
\end{figure}

We will describe the operator explicitly in the case where the first derivative is separated.

Let $Y\subseteq\C^2$ be a subspace, $U\colon \ell^2(\mE_-,\alpha_-)\to \ell^2(\mE_+,\alpha_+)$ linear, i.e.\ $U\in\C$.
Note that $U$ is contractive if and only if $\abs{U}^2\alpha_+\leq \alpha_-$.

\begin{example}\label{ex:twouncoupled}
  Let $Y:=\{(0,0)\}$. Then
  \begin{align*}
    D(A_{Y,U}) & = \{u\in D(A_0^*);\; u(0-) = u(0+) = 0, \, u'(0+) = U u'(0-)\},\\
    A_{Y,U} u & = -A_0^*u.
  \end{align*}
  By Corollary \ref{cor:skew-self-adjoint_separated}, $A_{Y,U}$ generates a unitary group provided $\abs{U}^2 = \frac{\alpha_-}{\alpha_+}$.
  By Theorem \ref{thm:contraction_separated}, $A_{Y,U}$ generates a semigroup of contractions provided $\abs{U}^2 \leq \frac{\alpha_-}{\alpha_+}$. Observe that if we take $U=0$, we are effectively reducing the Airy equation on the star graph $\mG$ to decoupled Airy equations on two halflines $(0,\infty)$ and $(-\infty,0)$ with Dirichlet conditions (for both equations) and Neumann (on the positive halfline only) boundary conditions. 
  The Airy equation on either of these halflines with said boundary conditions  has been considered often in the literature, see e.g.~\cite{Holmer06,FokasHiMan16}.
\end{example}

\begin{example}
  Let $Y:=\lin\{(0,1)\}$. Then
  \begin{align*}
    D(A_{Y,U}) & = \{u\in D(A_0^*);\; u(0-) = 0,\, u''(0+)\alpha_+ = -\frac{\beta_+}{2} u(0+),\, u'(0+) = U u'(0-)\},\\
    A_{Y,U} u & = -A_0^*u.
  \end{align*}
If $U=0$, these transmission conditions can be interpreted as a reduction of the system to two decoupled halflines: a Dirichlet condition is imposed on one of them, while a transmission condition that is the third-order counterpart of a Robin condition is imposed on the other one, along with a classical Neumann condition.

  By Corollary \ref{cor:skew-self-adjoint_separated}, $A_{Y,U}$ generates a unitary group provided $\abs{U}^2 = \frac{\alpha_-}{\alpha_+}$.
  By Theorem \ref{thm:contraction_separated}, $A_{Y,U}$ generates a semigroup of contractions provided $\abs{U}^2 \leq \frac{\alpha_-}{\alpha_+}$.
\end{example}

\begin{example}
  Let $Y:=\lin\{(1,0)\}$. Then
  \begin{align*}
    D(A_{Y,U}) & = \{u\in D(A_0^*);\; u(0+) = 0,\, u''(0-)\alpha_- = -\frac{\beta_-}{2} u(0-),\, u'(0+) = U u'(0-)\},\\
    A_{Y,U} u & = -A_0^*u.
  \end{align*}
If $U=0$, these transmission conditions can be interpreted as a reduction of the system to two decoupled halflines: Dirichlet and Neumann conditions are imposed on one of them, while the analog of a Robin condition is imposed on the other one.
  By Corollary \ref{cor:skew-self-adjoint_separated}, $A_{Y,U}$ generates a unitary group provided $\abs{U}^2 = \frac{\alpha_-}{\alpha_+}$.
  By Theorem \ref{thm:contraction_separated}, $A_{Y,U}$ generates a semigroup of contractions provided $\abs{U}^2 \leq \frac{\alpha_-}{\alpha_+}$.
\end{example}

Due to lack of conditions on $u''(0+)$ and/or $u''(0-)$,~\eqref{eq:yesmass} cannot be generally satisfied in any of the previous three cases and therefore the corresponding systems do not enjoy conservation of mass.

\begin{example}
  Let $Y:=\lin\{(1,1)\}$. Then
  \begin{align*}
    D(A_{Y,U}) & = \{u\in D(A_0^*);\; u(0-) = u(0+)=:u(0),\, u''(0+)\alpha_+ - u''(0-)\alpha_- = \frac{\beta_- - \beta_+}{2} u(0),\\
    & \qquad\qquad\qquad\quad u'(0+) = U u'(0-)\},\\
    A_{Y,U} u & = -A_0^*u.
  \end{align*}
  By Corollary \ref{cor:skew-self-adjoint_separated}, $A_{Y,U}$ generates a unitary group provided $\abs{U}^2 = \frac{\alpha_-}{\alpha_+}$: observe that this is in particular the case if $\alpha_+=\alpha_-$, $\beta_+=\beta_-$ and $U=1$, meaning that not only $u$, but also $u'$ and $u''$ are continuous in the origin: this is the classical case considered in the literature and amounts to the free Airy equation on $\mathbb R$, see e.g.\ the summary in~\cite[\S~7.1]{LinPon14}. 
  In view of Remark~\ref{rem:mass}, generation of a mass-preserving unitary group still holds under the more general assumption that $U=e^{i\phi}$ for some $\phi \in [0,2\pi)$, while in view of the prescribed transmission conditions~\eqref{eq:yesmass} cannot be satisfied unless $\beta_+=\beta_-$.
Even upon dropping the assumption that $\beta_+=\beta_-$ we obtain a the third order counterpart of a $\delta$-interaction, under which generation of a unitary group is still given.

On the other hand, by Theorem \ref{thm:contraction_separated}, $A_{Y,U}$ generates a semigroup of contractions already under the weaker assumption that $\abs{U}^2 \leq \frac{\alpha_-}{\alpha_+}$.
\end{example}


\begin{example}
  Let $Y:=\lin\{(1,-1)\}$. Then
  \begin{align*}
    D(A_{Y,U}) & = \{u\in D(A_0^*);\; u(0-) = - u(0+),\, u''(0+)\alpha_+ + u''(0-)\alpha_- = \frac{\beta_+ - \beta_-}{2} u(0-),\, u'(0+) = U u'(0-)\},\\
    A_{Y,U} u & = -A_0^*u.
  \end{align*}
  By Corollary \ref{cor:skew-self-adjoint_separated}, $A_{Y,U}$ generates a unitary group provided $\abs{U}^2 = \frac{\alpha_-}{\alpha_+}$: we can regard this case as a third-order counterpart of $\delta'$-interactions  of second-order operators. By~\eqref{eq:yesmass}, the system enjoys conservation of mass if and only if additionally
  \begin{equation}\label{eq:yesmass-a2}
2\alpha_+ u''(0+)+\frac{\beta_+ - \beta_-}{2} u(0-)=0\ ,
\end{equation}
which is generally not satisfied (remember that we are not considering the degenerate case $\alpha=0$).
  
By Theorem \ref{thm:contraction_separated}, $A_{Y,U}$ generates a semigroup of contractions provided $\abs{U}^2 \leq \frac{\alpha_-}{\alpha_+}$.
\end{example}

\begin{example}
  Let $Y:=\C^2$. Then
  \begin{align*}
    D(A_{Y,U}) & = \{u\in D(A_0^*);\;u''(0-)\alpha_- = -\frac{\beta_-}{2}u(0-),\, u''(0+)\alpha_+ = -\frac{\beta_+}{2} u(0+),\, u'(0+) = U u'(0-)\},\\
    A_{Y,U} u & = -A_0^*u.
  \end{align*}
  These transmission conditions amount to considering two decoupled systems,
each with Robin-like conditions along with a Neumann condition on one of them.
  By Corollary \ref{cor:skew-self-adjoint_separated}, $A_{Y,U}$ generates a unitary group provided $\abs{U}^2 = \frac{\alpha_-}{\alpha_+}$; by~\eqref{eq:yesmass}, the system enjoys conservation of mass if and only if additionally
  \begin{equation}\label{eq:yesmass-a1}
-\frac{\beta_-}{2}u(0-)+\frac{\beta_+}{2} u(0+)+\beta_-u(0-)-\beta_+ u(0+)=0\ ,
\end{equation}
which is generally only satisfied if $\beta_-=\beta_+=0$.
  By Theorem \ref{thm:contraction_separated}, $A_{Y,U}$ generates a semigroup of contractions provided $\abs{U}^2 \leq \frac{\alpha_-}{\alpha_+}$.
\end{example}

Let us now describe one particular example where the first derivative is not separated and the corresponding semigroup is unitary.
\begin{example}
  We consider the case $\alpha_- = \alpha_+ = 1$, $\beta_- = \beta_+ = 0$.
  Define the $3\times 3$-matrix $L:\mathcal G_-\to \mathcal G_+$ by
  \[L = \begin{pmatrix}
	  1 & 0 & 0 \\
	  \sqrt{2} & 1 & 0 \\
	  1 & \sqrt{2} & 1
        \end{pmatrix}.\]
  An easy calculation yields
  \[L^* \begin{pmatrix}
	  0 & 0 & -1 \\
	  0 & 1 & 0 \\
	  -1 & 0 & 0 
        \end{pmatrix} L = \begin{pmatrix}
	  0 & 0 & -1 \\
	  0 & 1 & 0 \\
	  -1 & 0 & 0 
        \end{pmatrix},\]
  i.e.\ $L^*B_+ L = B_-$. Thus, for $x,y\in \mathcal{G}_-$ we have
  \[\langle Lx \mid Ly\rangle_+ = \bigl(B_+Lx \mid Ly\bigr) = \bigl(L^*B_+L x\mid y \bigr) = \bigl( B_- x \mid y\bigr) = \langle x \mid y\rangle_-.\]
  Hence, $L$ defines a $(\mathcal{G}_-,\mathcal{G}_+)$-unitary operator. Combining Theorem \ref{thm:characselforthog} and Theorem \ref{thm:characskewself} we obtain that $A_L$ is skew-self-adjoint.
  Note that
  \begin{align*}
    D(A_L) = \{u\in D(A_0^*);\; & u(0+) = u(0-)=:u(0),\, u'(0+) = \sqrt{2} u(0) + u'(0-),\\
    & u''(0+) = u(0) + \sqrt{2} u'(0-) + u''(0-)\},
  \end{align*}
  so the transmission conditions couple the values of the function, its first and second derivative at the boundary point. By~\eqref{eq:yesmass}, mass is not conserved.\\
 \begin{remark}
As mentioned in the Introduction, in the recent paper ~\cite{DeconinckSheilsSmith16} the Airy equation is treated (with $\beta=0$, which makes difference when translation invariance is broken) on a line with an interface at a point where  linear transmission conditions are imposed. The authors give conditions for the solvability of  the evolution  problem in terms of the coefficients appearing in the transmission conditions. An interesting problem could be to compare the conditions there obtained with the ones given in the present paper in the case of two half-lines. 
 \end{remark}
\end{example}

\subsection{The case of three halflines}

Let us now describe the operator (again with separated first derivative) for the case of three halflines.
Let $\abs{\mE_-} = 1$, $\abs{\mE_+} = 2$, which describes a branching channel. By Proposition \ref{prop:first} it is impossible that the Airy equation is governed by a unitary group in this setting; however, we are going to discuss a few concrete cases where a semigroup of contractions is generated by the Airy operator.

Let $Y$ be a subspace of $\ell^2(\mE_-)\oplus \ell_+(\mE_+) \cong \C\oplus \C^2 \cong \C^3$ and $U\colon \ell^2(\mE_-,\alpha_-)\to \ell^2(\mE_+,\alpha_+)$ be a linear mapping, i.e.\ $U\in M_{21}(\C)\simeq \C^2$: we denote for simplicity
\[
U=(U_1,U_2)\ .
\]
Then $U$ is -contractive if and only if $\abs{U_1}^2\alpha_{+,1} + \abs{U_2}^2\alpha_{+,2} \leq \alpha_-$.
Note that $U^*$ is given by $U^*_1 = \overline{U_1}\frac{\alpha_{+,1}}{\alpha_-}$, $U^*_2 = \overline{U_2}\frac{\alpha_{+,2}}{\alpha_-}$.
\begin{example}
  Let $Y:=\{(0,0,0)\}$. Then
  \begin{align*}
    D(A_{Y,U}) & = \{u\in D(A_0^*);\; u(0-) = u_1(0+) = u_2(0+) = 0, \, (u_1'(0+),u_2(0+))^\top = u'(0-) U\},\\
    A_{Y,U} u & = -A_0^*u.
  \end{align*}
As in Example~\ref{ex:twouncoupled}, whenever $U=(0,0)$ the associated system effectively reduces to three decoupled halflines with Dirichlet, resp.\ Dirichlet and Neumann boundary conditions.
  By Theorem \ref{thm:contraction_separated}, $A_{Y,U}$ generates a semigroup of contractions provided 
  \[
 \abs{U_1}^2\alpha_{+,1} + \abs{U_2}^2\alpha_{+,2} \leq \alpha_-.
  \]
In view of~\eqref{eq:yesmass}, mass is generally not conserved regardless of $U$.
\end{example}

\begin{example}
  Let $Y:=\lin\{(1,1,1)\}$. Then
  \begin{align*}
    D(A_{Y,U}) & = \{u\in D(A_0^*);\; u(0-) = u_1(0+) = u_2(0+)=:u(0),\\
    & \qquad\qquad\qquad -\alpha_- u''(0-) + \alpha_{+,1} u_1''(0+) + \alpha_{+,2}u_2''(0+) = \frac{\beta_- - \beta_{+,1}-\beta_{+,2}}{2} u(0),\\
    & \qquad\qquad\qquad (u_1'(0+),u_2(0+))^\top = u'(0-) U\},\\
    A_{Y,U} u & = -A_0^*u.
  \end{align*}
  Observe that the considered transmission conditions impose continuity of the values of $u$ in the center of the star; in fact, the transmission conditions are the analog of a  $\delta$-interaction.
  By Theorem \ref{thm:contraction_separated}, $A_{Y,U}$ generates a semigroup of contractions provided 
 \[
 \abs{U_1}^2\alpha_{+,1} + \abs{U_2}^2\alpha_{+,2} \leq \alpha_-.
 \]
  Regardless of $U$, this semigroup is mass preserving if and only if $\beta_- - \beta_{+,1}-\beta_{+,2}=0$.
\end{example}

\vskip10pt

\vskip5pt {\bf Acknowledgements.}
\par\noindent 
The authors are grateful to Roberto Camassa for a useful discussion.

\end{document}